\documentclass[a4paper,12pt]{article}

\usepackage[utf8]{inputenc}
\usepackage[english]{babel}		% set languages - the last one is the first used language
\usepackage{csquotes} % quotation marks (can run into conflict for german & english)

\usepackage[backend=biber,sorting=nyt,style=numeric-comp,doi=false,date=year,isbn=false,maxbibnames=9,url=false,firstinits=true,bibwarn=true]{biblatex} % style, citestyle
\usepackage[left=2.5cm,right=2.5cm,top=2.5cm,bottom=2.5cm]{geometry} % option showframe, other possible options: bindingoffset=10mm,centering,includeheadfoot,margin=3cm
\usepackage{wrapfig}
\usepackage[super]{nth}

\usepackage{amssymb}
\usepackage{amsmath}
\usepackage{amsfonts}   
\usepackage{mathtools}
\usepackage{mathrsfs}				% \mathscr, for example the Lagrangiandensity
\usepackage{amsbsy}
\usepackage{faktor}                 % for factor/quotient rings
\usepackage{dsfont}				% for \mathds{}, in texlive-fonts-extra, e.g. \mathds{1}
\usepackage{amsthm}				% theorems
\newtheoremstyle{owntheorem}{}% measure of space to leave above the theorem. E.g.: 3pt
  {}% measure of space to leave below the theorem. E.g.: 3pt
  {\itshape}% name of font to use in the body of the theorem
  {}% measure of space to indent
  {}% name of head font
  {:}% punctuation between head and body
  {1em}% space after theorem head; " " = normal interword space
  {\textbf{\thmname{#1}\thmnumber{ #2}}\thmnote{ [#3]}}% Manually specify head
\newtheoremstyle{owndefinition}%
  {}{}%
  {}{}%
  {}{:}%
  {1em}%
  {\textbf{\thmname{#1}\thmnumber{ #2}}\thmnote{ [#3]}}
\theoremstyle{owntheorem}
\newtheorem{theorem}{Theorem}[section]
\newtheorem{lemma}[theorem]{Lemma}
\newtheorem{cor}[theorem]{Corollary}

\theoremstyle{owndefinition}

\newtheorem{example}{Example}[section]
\newtheorem*{remark}{Remark}
\newtheorem*{acknowledgements}{Acknowledgements}

\usepackage{url} 				% formats URLs with \url{}
\usepackage{wasysym}				% symbols e.g. \checked
\usepackage{enumerate}
\usepackage{rotating}				% for sidewaystable
\usepackage{multirow}
\usepackage{array}
\usepackage{longtable}
\usepackage{lscape}
\usepackage{fancyvrb}               % frames for verbatim
\usepackage{tikz}
\usepackage{tablefootnote}
\usepackage{wasysym}
\usepackage{blkarray}
\usepackage{authblk}
\usepackage{spverbatim}
\numberwithin{equation}{section} % change numeration of equations 

\newcommand{\diag}{\operatorname{diag}}

\newcommand{\ch}{\operatorname{char}}

\newcommand{\Arg}{\operatorname{Arg}}
\renewcommand\Re{\operatorname{Re}} 		% real part as operator Re
\newcommand{\pd}[2]{\frac{\partial #1}{\partial #2}}

\newcommand{\comma}{\quad \textrm{ ,}} % comma in math mode
\newcommand{\point}{\quad \textrm{ .}} % point in math mode
\newcommand{\monthword}[1]{\ifcase#1\or \or \or \or April\fi}
\newcommand{\vol}{\operatorname{vol}}

\newcommand{\rank}{\operatorname{rank}}
\newcommand{\Aa}{\mathcal{A}}

\newcommand{\nuu}{{\underline{\nu}}}

\newcommand{\Conv}{\operatorname{Conv}}
\newcommand{\relint}{\operatorname{relint}}

\newcommand{\Adj}{\operatorname{Adj}}

\newcommand{\Newt}{\operatorname{Newt}}
\newcommand{\Cc}{\overline{\mathcal C}_f^{\mathrm c}}
\newcommand{\Aff}{\operatorname{Aff}}
\newcommand{\AnK}{\mathbb{A}^n_{\mathbb{K}}}
\newcommand{\PnK}{\mathbb{P}^n_{\mathbb{K}}}
\newcommand{\Sing}{\operatorname{Sing}}
\newcommand{\codim}{\operatorname{codim}}
\newcommand{\Uu}{\mathcal{U}}
\newcommand{\Ff}{\mathcal{F}}
\newcommand{\Gg}{\mathcal{G}}
\newcommand{\Csn}{(\mathbb C^*)^n}

\newcommand{\CNV}{completely non-vanishing }
\newcommand{\HKP}{Horn-Kapranov-parametrization }
\newcommand{\HKp}{Horn-Kapranov-parametrization}

\title{Kinematic singularities of Feynman integrals and principal $A$-determinants}
\author{René Pascal Klausen\footnote{Department of Physics at Johannes-Gutenberg University of Mainz and Humboldt University of Berlin (klausen@physik.hu-berlin.de)}}
\date{\today}

\begin{document}
  \maketitle
  
  \begin{abstract}
    We consider the analytic properties of Feynman integrals from the perspective of general $A$-discriminants and $\Aa$-hypergeometric functions introduced by Gelfand, Kapranov and Zelevinsky (GKZ). This enables us, to give a clear and mathematically rigour description of the singular locus, also known as Landau variety, via principal $A$-determinants. We also comprise a description of the various second type singularities. Moreover, by the \HKP we give a very efficient way to calculate a parametrization of Landau varieties. We furthermore present a new approach to study the sheet structure of multivalued Feynman integrals by use of coamoebas.
  \end{abstract}

  \section{Introduction}
  
    Feynman integrals are usually understood as functions depending on various observables and parameters. Even though the physical observables do not take complex values in measurements, these Feynman integrals can only be thought consistently in complex domains. By considering Feynman integrals as complex functions and examining their analytic properties, surprising connections were found, for example dispersion relations and Cutkosky's rules \cite{CutkoskySingularitiesDiscontinuitiesFeynman1960, BlochCutkoskyRulesOuter2015}. As conjectured for the first time by T. Regge, it seems that these connections are not just arbitrary and indicate a more fundamental relation between the monodromy group and the fundamental group for Feynman integrals in the context of Picard-Lefschetz theory (see e.g. \cite{SpeerGenericFeynmanAmplitudes1971, PonzanoMonodromyRingsClass1969} for Regge's perspective). Apart from these conceptual questions, the analytic structure plays also an important role in many practical approaches, for example the analytic continuation, Steinman relations \cite{CahillOpticalTheoremsSteinmann1975}, sector decomposition \cite{BinothNumericalEvaluationMultiloop2004, AnastasiouEvaluatingMultiloopFeynman2007, BorowkaNumericalEvaluationMultiloop2013} or certain methods in QCD \cite{LibbyMassDivergencesTwoparticle1978}.

    Thus, there are many good reasons to assume that one needs to understand the analytic structure in order to obtain a profound insight of perturbative scattering amplitudes. However, compared with the importance of this subject, little is known about the analytic structure of Feynman integrals in general. A main reason for this gap seems to be the sophisticated nature of the mathematical framework. Furthermore, the analytic structure is also a hard problem from an algorithmical point of view: only for a few simple Feynman graphs the analytic structure is known.
    
    The structural investigation of the analytic properties of Feynman integrals was started in 1959 independently by Bjorken \cite{BjorkenExperimentalTestsQuantum1959}, Nakanishi \cite{NakanishiOrdinaryAnomalousThresholds1959} and Landau \cite{LandauAnalyticPropertiesVertex1959} and is also known as Landau analysis. Since a comprehensive historical review on Landau analysis can be found in \cite{MizeraCrossingSymmetryPlanar2021}, we restrict ourselves to a very short historical overview. For a summary of the first steps of this subject from the 1960s we refer to the monograph of Eden et al. \cite{EdenAnalyticSmatrix1966}. A much more mathematically profound investigation was carried out by Pham et al. in terms of homology theory \cite{HwaHomologyFeynmanIntegrals1966, PhamSingularitiesIntegrals2011}. Pham's techniques have recently brought back into focus by S. Bloch and D. Kreimer \cite{BlochCutkoskyRulesOuter2015}. An alternative approach avoiding the introduction of homology theory was initiated by Regge, Ponzano, Speer and Westwater \cite{PonzanoMonodromyRingsClass1969, SpeerGenericFeynmanAmplitudes1971}. Their work was also the starting point for a mathematical treatment due to Kashiwara and Kawai \cite{KashiwaraConjectureReggeSato1976}, Sato \cite{SatoRecentDevolpmentHyperfunction1975} and Sato et al. \cite{SatoHolonomyStructureLandau1977}, which are all heavily based on partial differential equation systems. Currently, there is a renewed interested in Landau varieties and we refer to \cite{CollinsNewCompleteProof2020, MuhlbauerMomentumSpaceLandau2020} for a selection of modern approaches as well as \cite{BonischAnalyticStructureAll2021, BonischFeynmanIntegralsDimensional2021}, where the analytic structure of specific Feynman integrals was studied in the context of differential equations by methods of topological string theory on Calabi-Yau manifolds. \\
    
    In this work we will also use partial differential equations in order to analyze the singular locus of Feynman integrals and we will do this in the framework of $\Aa$-hypergeometric functions, which were introduced by Gelfand, Kapranov, Zelevinsky (GKZ) and collaborators \cite{GelfandEquationsHypergeometricType1988, GelfandHypergeometricFunctionsToral1989, GelfandGeneralizedEulerIntegrals1990, GelfandHypergeometricFunctionsToric1991, GelfandGeneralHypergeometricSystems1992}. These $\Aa$-hypergeometric functions, where Feynman integrals are a special case, are defined as solutions of partial differential equation systems and their singular locus can be obtained by the characteristic ideal of that system. This approach gives us very familiar equations known from Landau analysis. However, as the $\Aa$-hypergeometric theory is well understood from a conceptual point of view, we can benefit from it for Feynman integrals, inter alia we get very useful tools as the \HKp.   
    
    We will start this article with a short review of $A$-resultants and $A$-discriminants, which state multivariate generalizations of the usual univariate resultants and discriminants. In particular we will also introduce the principal $A$-determinant, which is a specific $A$-resultant. In the following section \ref{sec:AHyp} we will define the $\Aa$-hypergeometric systems and we will reproduce the connection between the singular locus of $\Aa$-hypergeometric functions and convenient principal $A$-determinants. Afterwards we will apply these methods to Feynman integrals in order to obtain a rigorous description of the singular locus, i.e. the Landau variety, by means of principal $A$-determinants. Those principal $A$-determinants factorize in several $A$-discriminants and we can detect factors which correspond to second-type singularities as well as parts corresponding to normal, anomalous and pseudo thresholds. In section \ref{sec:coamoebas} we will introduce the concept of coamoebas, which allow a more detailed examination of the structure of the singular locus. In particular the coamoeba provides a condition to detect pseudo thresholds. In order to demonstrate the advantage of the \HKP we will finish this article by giving a parametrization of the leading Landau singularity of the double-edged triangle graph (also known as dunce's cap) in section \ref{sec:example} as an example.

    \begin{acknowledgements}
      I would like to express my special thanks towards Christian Bogner for his personal support on all questions. My gratitude also goes to Dirk Kreimer and his group for hospitality and helpful discussions. I also benefit from correspondence with Erik Panzer, Marko Berghoff and Max Mühlbauer. This research is supported by the cluster of excellence ``Precision Physics, Fundamental Interactions and Structure of Matter'' (PRISMA+) at Johannes-Gutenberg university of Mainz.
    \end{acknowledgements}

  \section{$A$-resultants and $A$-discriminants} \label{sec:AResultantsDiscriminants}
  
    We are often interested in the question whether a system of simultaneous polynomial equations 
    \begin{align}
      f_0  (x_1,\ldots,x_n) = \ldots = f_n  (x_1,\ldots,x_n) = 0 \label{eq:polysys}
    \end{align}
    has a solution in a given (algebraically closed) field $\mathbb K$ or if it is inconsistent. That question could be answered in general by calculating the Groebner basis of the ideal generated by $f_0,\ldots,f_n$. Thus, a system of polynomial equations is inconsistent if and only if the corresponding reduced Groebner basis is equal to $1$. Unfortunately, the calculation of Groebner bases can be hopeless complicated and computers fail even in simpler examples. Resultants, instead, can answer this question much more efficiently. In general a resultant is a polynomial in the coefficients of the polynomials $f_0,\ldots,f_n$, which vanishes whenever the system (\ref{eq:polysys}) has a common solution. 
    
    However, the theory of multivariate resultants comes with several subtleties. We have to distinguish between \textit{classical multivariate resultants} (also known as \textit{dense resultants}) and \textit{(mixed) $A$-resultants} (or \textit{sparse resultants}). The classical multivariate resultant will be applied to $n$ homogeneous polynomials in $n$ variables, where every polynomial consists in all possible monomials of a given degree, and detects common solutions in projective space $\mathbb P^{n-1}(\mathbb K)$. In contrast, the $A$-resultant is usually used, if the polynomials do not consist in all monomials of a given degree. For a system of $n+1$ polynomials in $n$ variables, the $A$-resultant is a custom-made polynomial and reveals common solutions, which are mostly located in the affine space $(\mathbb C^*)^n$. However, what we accept as a ``solution'' in the latter case is slightly subtle and we will give a precise definition below. Note, that the classical multivariate resultant is a special case of the $A$-resultant \cite{CoxUsingAlgebraicGeometry2005}.
    Furthermore, we want to distinguish between the case where all polynomials $f_0,\ldots,f_n$ having the same monomial structure, which will be defined by a single set $A$ and the mixed case where the polynomials $f_0,\ldots,f_n$ having different monomial structure defined by several sets $A_0,\ldots,A_n$. 
   
    Closely related to resultants are discriminants, which determine whether a polynomial $f$ has a multiple root. This is equivalent to ask if there is a solution such that the polynomial $f$ and its first derivatives vanish. Hence, discriminants play also an important role for identifying singular points of algebraic varieties.
    
    In the following we will sketch several key features of the theory of $A$-resultants and $A$-discriminants, which were mainly introduced in a series of articles by Gelfand, Kapranov and Zelevinsky \cite{GelfandAdiscriminantsCayleyKoszulComplexes1990, GelfandDiscriminantsPolynomialsMany1990, GelfandNewtonPolytopesClassical1990, GelfandDiscriminantsPolynomialsSeveral1991} in the study of $A$-hypergeometric functions \cite{GelfandHypergeometricFunctionsToral1989, GelfandGeneralizedEulerIntegrals1990, GelfandGeneralHypergeometricSystems1992} and were collected in \cite{GelfandDiscriminantsResultantsMultidimensional1994}. For an introduction to $A$-resultants as well as the classical multivariate resultants we refer to \cite{CoxUsingAlgebraicGeometry2005} and \cite{SturmfelsSolvingSystemsPolynomial2002}.
    
    \subsection{Basic notions}
    
      Before introducing the generalized multivariate resultants and discriminants, let us first recall some basic notions in the language of polynomials and algebraic geometry, which can also be found in the most textbooks e.g. \cite{HartshorneAlgebraicGeometry1977, HolmeRoyalRoadAlgebraic2012}. \\
      
      Let $\mathbb K$ be an algebraically closed field, typically the complex numbers $\mathbb C$ or a convenient subfield. The \textit{affine $n$-space} $\AnK$ over $\mathbb K$ is the set of $n$-tuples of elements from $\mathbb K$, where we adopt some structure from the vector space $\mathbb K^n$ in order to treat relative positions of points in $\AnK$. Thus, we assume the existence of a map $S: \AnK \times \AnK \rightarrow \mathbb K^n$, which relates two points from $\AnK$ with a vector in $\mathbb K^n$ representing the relation between these points. For points $a,b,c\in \AnK$ we claim that $S(a,b)+S(b,c) = S(a,c)$ and $S(a,a)=0$. Furthermore we want $S(a,\cdot)$ to be bijective for every $a\in \AnK$. Hence, an affine $n$-space over a field $\mathbb K$ is the vector space $\mathbb K^n$ but with different morphisms.\\
      
      By $\mathbb K[x_1,\ldots,x_n]$ we denote the coordinate ring of $\AnK$, i.e. the ring of polynomials in the variables $x_1,\ldots,x_n$. Furthermore, we can introduce a topology on $\AnK$, the so-called \textit{Zariski topology}, by defining the zero locus of polynomials to be the closed sets. For a set of polynomials $f_1,\ldots,f_k\in\mathbb K[x_1,\ldots,x_n]$ we call
      \begin{align}
        \mathbf V(f_1,\ldots,f_k) := \left\{ x\in \AnK \, \rvert\, f_1(x) = \ldots = f_k(x) = 0\right\} \subseteq \AnK
      \end{align}
      an \textit{affine variety} generated by $f_1,\ldots,f_k$. This definition extends naturally to an ideal of polynomials. An affine variety is called \textit{irreducible}, if it can not be written as the union of two proper subvarieties. \\
      
      Laurent polynomials can be treated in an analogous way. We call the affine variety $\Csn$ the \textit{algebraic torus} and \textit{Laurent monomials} in the variables $x_1,\ldots,x_n$ are nothing else than the characters of $\Csn$
      \begin{align}
        (\mathbb C^*)^n \rightarrow \mathbb C^*, \qquad x \mapsto x^a := x_1^{a_1} \cdots x_n^{a_n} \label{eq:LaurentMonomial}
      \end{align}
      with the exponent $a=(a_1,\ldots,a_n)\in\mathbb Z^n$. Here and in the following we will make use of a multiindex notation as indicated in (\ref{eq:LaurentMonomial}). A \textit{Laurent polynomial} is a finite linear combination of these monomials and can be uniquely written as  
      \begin{align}
        f(x)=f(x_1,\ldots,x_n) = \sum_{a\in A} z_a x^a \in \mathbb C[x_1^{\pm 1},\ldots,x_n^{\pm 1}]
      \end{align}
      where $A\subset \mathbb Z^n$ is a finite set of non-repeating points and we consider the coefficients $z_a\in\mathbb C$ to be complex numbers not identically zero. We will call $A$ the \textit{support} of $f$. The space of all Laurent polynomials with fixed monomials from $A$ but with indeterminate coefficients $z_a$ will be denoted as $\mathbb C^A$. Thus, the set of coefficients $\{z_a\}_{a\in A}$ are coordinates of $\mathbb C^A$.\\
    
      For elements $a^{(1)},\ldots,a^{(k)}\in \AnK$ of an affine space we call
      \begin{align}
        \lambda_1 a^{(1)} + \ldots + \lambda_k a^{(k)} \qquad \text{with}\qquad \sum_{i=1}^k \lambda_i = 1, \quad \lambda_i\in\mathbb K
      \end{align}
      an \textit{affine combination}. Similarly, for a subset $S\subseteq \mathbb K$ we denote by
      \begin{align}
        \Aff_S \left(a^{(1)},\ldots,a^{(k)}\right) := \left\{\lambda_1 a^{(1)} + \ldots + \lambda_k a^{(k)} \,\Big\rvert\, \lambda_i \in  S, \sum_{i=1}^k \lambda_i = 1\right\}
      \end{align}
      the \textit{affine span} generated by the elements $a^{(1)},\ldots,a^{(k)}\in\AnK$ over $S$.  A discrete subgroup of an affine space $\AnK$ is called an \textit{affine lattice} if the subgroup spans the full space $\AnK$. Further, a map $f: \AnK \rightarrow \mathbb A_{\mathbb K}^m$ between two affine spaces is called \textit{affine transformation}, if it preserves all affine combinations.\\
      
      A  minimal generating set of elements $a^{(0)},\ldots,a^{(n)}$ which spans the whole affine space $\AnK$ over $\mathbb K$, is called a \textit{basis} of an affine space $\AnK$ (or a \textit{barycentric frame}). Thus, for a given basis $a^{(0)},\ldots,a^{(n)}$ we can write every element $a\in \AnK$ as an affine combination of that basis and we will call the corresponding tuple $(\lambda_0,\ldots,\lambda_{n})$ the \textit{barycentric coordinates} of $a$. \\
      
      These barycentric coordinates indicate that we can naturally identify the affine space $\AnK$ as a hyperplane in the vector space $\mathbb K^{n+1}$. Thus, we consider the vector space $\mathbb K^{n+1} = \mathbb K^n \cup (\mathbb K^* \times \AnK)$ consisting in the slice of the vector space $\mathbb K^n$, containing the origin, and the remaining slices, each corresponding to an affine spaces $\AnK$. Since all slices, which do not contain the origin behave the same, we can identify w.l.o.g. $\AnK$ as the slice $1\times \AnK$ of $\mathbb K^{n+1}$. This will enable us to treat affine objects with the methods of linear algebra. Therefore, we can accomplish the embedding, by adding an extra coordinate
      \begin{align}
        a = (a_1,\ldots,a_n) \mapsto (1,a_1,\ldots,a_n) \label{eq:homogenization} \point
      \end{align}
      Since points lying on a common hyperplane of $\mathbb K^{n+1}$, correspond to exponents of quasi-homogeneous polynomials, we will call the map of (\ref{eq:homogenization}) \textit{homogenization}. For a finite subset of lattice points $A=\{a^{(1)},\ldots,a^{(N)}\}\subset \mathbb Z^n$, we will write $\mathcal A = \{(1,a^{(1)}),\ldots,(1,a^{(N)})\} \subset \mathbb Z^{n+1}$ as its homogenized version. Whenever it is convenient, we will denote by $\Aa$ also the $(n+1)\times N$ integer matrix collecting the elements of the subset $\Aa\subset\mathbb Z^{n+1}$ as columns.\\
      
      Furthermore, by $\mathbb L := \ker_{\mathbb Z}(\Aa)\subseteq\mathbb Z^N$ we denote the integer kernel of the homogenized matrix $\Aa\in\mathbb Z^{(n+1)\times N}$. The lattice $\mathbb L$ has rank $r:=N-n-1$ and every basis $\mathcal B=\{b_1,\ldots,b_r\}\subset\mathbb Z^N$ of that lattice will be called a \textit{Gale dual} of $\Aa$. Analogue to $\Aa$, we will consider $\mathcal B$ also as an $N\times r$ integer matrix whenever it is convenient.\\
       
      Closely related to the affine space is the projective space. The \textit{projective space} $\PnK$ is the set of equivalence classes in $\mathbb K^{n+1}\setminus\{0\}$, where two elements $a,b\in \mathbb K^{n+1}\setminus\{0\}$ are equivalent if there exists a number $\lambda\in\mathbb K^*$ such that $a=\lambda b$. Thus, points of the projective space can be described by homogeneous coordinates. A point $a\in\PnK$ is associated to the homogeneous coordinates $[s_0:\ldots :s_n]$ if an arbitrary element of the equivalence class of $a$ is described in the vector space $\mathbb K^{n+1}$ by the coordinates $(s_0,\ldots,s_n)$. Note that the homogeneous coordinates are not unique, as they can be multiplied by any element $\lambda\in\mathbb K^*$. 
    
      Furthermore, we can decompose the projective space $\PnK= \AnK \cup \mathbb P^{n-1}_{\mathbb K}$ into an affine space and a projective space of lower dimension. In coordinates that decomposition means, that in case of $s_0\neq 0$, we can choose w.l.o.g. $s_0=1$, which defines the aforementioned affine hyperplane in $\mathbb K^{n+1}$. For $s_0=0$, sometimes referred as the ``points at infinity'' due to $\left[1 : \frac{s_1}{s_0} : \ldots : \frac{s_n}{s_0} \right]$, we obtain the projective space of lower dimension by the other remaining coordinates $[s_1 : \ldots : s_n]$. \\
     
      For a finite subset of points $A=\{a^{(1)},\ldots,a^{(N)}\}\subset\mathbb Z^n$ we define a \textit{(lattice) polytope} $P$ as the convex hull of these points
      \begin{align}
        P := \Conv (A) = \left\{ \lambda_1 a^{(1)} + \ldots + \lambda_N a^{(N)} \,\Big\rvert\, \lambda_i \in\mathbb R, \lambda_i \geq 0, \sum_{i=1}^N \lambda_i = 1 \right\} \subset \mathbb R^n \point
      \end{align}
      Alternatively, we can describe every polytope uniquely as a bounded, minimal intersection of half-spaces
      \begin{align}
        P := P(M,b) = \left\{ \mu\in\mathbb R^n \,\rvert\, m_j^T \cdot \mu \leq b_j, 1 \leq j \leq k \right\} \label{eq:HPolytope}
      \end{align}
      where $m_j \in \mathbb Z^N$ and $b_j\in\mathbb Z$ are relatively prime for $j = 1,\ldots,k$. 
      
      Especially, for polynomials $f(x) = \sum_{a\in A} z_a x^a\in\mathbb C[x_1^{\pm 1},\ldots,x_n^{\pm 1}]$ we define the \textit{Newton polytope}
      \begin{align}
        \Newt (f) := \Conv \left(\left\{a\in A \,\rvert\, z_a\not\equiv 0\right\}\right) \subset \mathbb R^n
      \end{align}
      as the convex hull of the exponent vectors. A subset $\tau\subseteq P$ of a polytope for which there exists a linear map $\phi : \mathbb R^n \rightarrow \mathbb R$, such that
      \begin{align}
        \tau = \left\{ p \in P | \phi(p) \geq \phi(q) \quad\text{for all } q\in P \right\} \subseteq P \label{eq:facedef}
      \end{align}
      is called a \textit{face} of $P$. Every face $\tau$ is itself a polytope, generated by a subset of $A$. Whenever it is convenient we will identify with $\tau$ also this subset of $A$, as well as the subset of $\{1,\ldots,N\}$ indexing the elements of $A$ corresponding to this subset. For a face $\tau\subseteq\Newt(f)$ of a Newton polytope, we define the \textit{truncated polynomial} with respect to $\tau$ as
      \begin{align}
        f_\tau (x) := \sum_{a\in A\cap \tau} z_a x^a
      \end{align}
      consisting only in the monomials corresponding to the face $\tau$.

      Further, we want to introduce a \textit{volume} $\vol(P)\in\mathbb Z_{\geq 0}$ for these lattice polytopes $P$, which is normalized such that the standard $n$-simplex (i.e. the convex hull of the standard unit vectors of $\mathbb R^{n}$ and the origin) has a volume equal to $1$.

    \subsection{Mixed $(A_0,\ldots,A_{n})$-resultants}
    
      After introducing the basic terms, we can now turn towards the multivariate resultants and discriminants. The key idea of resultants is to specify coefficients and variables in a system of polynomial equations and eliminate the variables from it. We will summarise the basic definitions and several properties of the multivariate resultants, which can be found in \cite{PedersenProductFormulasResultants1993, SturmfelsNewtonPolytopeResultant1994, SturmfelsIntroductionResultants1997, SturmfelsSolvingSystemsPolynomial2002, CoxUsingAlgebraicGeometry2005, GelfandDiscriminantsResultantsMultidimensional1994}. \\      
      
      Let $A_0,\ldots, A_{n} \subset \mathbb Z^n$ be finite subsets of the affine lattice $\mathbb Z^n$ and for every set $A_i$ we will consider the corresponding Laurent polynomial 
      \begin{align}
        f_i(x)=f_i(x_1,\ldots,x_n) = \sum_{a\in A_i} z^{(i)}_a x^a \in \mathbb C [x_1^{\pm 1}, \ldots, x_n^{\pm 1}]\point
      \end{align}
      For simplicity we will assume that the supports $A_0,\ldots,A_n$ jointly generate the affine lattice $\mathbb Z^n$. Furthermore, by $P_i := \Conv(A_i)$ we denote the Newton polytope of $f_i$. According to \cite{SturmfelsNewtonPolytopeResultant1994, SturmfelsIntroductionResultants1997, SturmfelsSolvingSystemsPolynomial2002}, we will call a configuration $A_0,\ldots,A_n$ \textit{essential} if
      \begin{align}
        \dim \left( \sum_{j=0}^n P_j \right) = n \qquad \text{and} \qquad  \dim \left(\sum_{j\in J} P_j\right) \geq  |J| \quad \text{for every } J \subsetneq\{0,\ldots,n\} \comma \label{eq:essential}
      \end{align}
      where the sum denotes the Minkowski sum and $|J|$ is the cardinality of a proper subset $J$. If all polytopes $P_i$ are $n$-dimensional, equations (\ref{eq:essential}) is trivially satisfied.
      
      In order to define the general resultants, we are interested in the set of coefficients $z_a^{(i)}$ for which there exists a solution of $f_0(x) = \ldots = f_n(x)=0$ for $x\in\Csn$. In other words we consider the following set in $\prod_{i=0}^n \mathbb C^{A_i}$
      \begin{align}
        \mathscr Z = \left\{ (f_0,\ldots,f_{n}) \in \prod_i \mathbb C^{A_i}| \mathbf V(f_0,\ldots,f_{n})\neq\emptyset \text{ in } (\mathbb C^*)^{n} \right\} \subseteq \mathbb \prod_{i=0}^n \mathbb C^{A_i} \point
      \end{align}
      Furthermore, by $\overline{\mathscr Z}$ we will denote the Zariski closure of $\mathscr Z$. The \textit{mixed $(A_0,\ldots,A_n)$-resultant} $R_{A_0,\ldots,A_n}(f_0,\ldots,f_n)\in\mathbb Z[\{z^{(i)}_a\}_{a\in A, i = 0,\ldots,n}]$ is an irreducible polynomial in the coefficients of the polynomials $f_0,\ldots,f_n$. In case where $\overline{\mathscr Z}$ describes a hypersurface in $\prod_i \mathbb C^{A_i}$ we will define $R_{A_0,\ldots,A_n}(f_0,\ldots,f_n)$ to be the (minimal) defining polynomial of this hypersurface $\overline{\mathscr Z}$. Otherwise, so if $\codim \overline{\mathscr Z} \geq 2$, we will set $R_{A_0,\ldots,A_n}(f_0,\ldots,f_n)=1$. The mixed $(A_0,\ldots,A_n)$-resultant always exists and is uniquely defined up to a sign, which was shown in \cite{GelfandDiscriminantsResultantsMultidimensional1994}.
      
      Further, we have $\codim \overline{\mathscr Z} = 1$ if and only if there exists a unique subset of $A_0,\ldots,A_n$ which is essential \cite{SturmfelsNewtonPolytopeResultant1994}. In that case the mixed $(A_0,\ldots,A_n)$-resultant coincides with the resultant of that essential subset.\\
      
      One has to remark as a warning, that the mixed $(A_0,\ldots,A_n)$-resultants not only detect common solutions of $f_0 = \ldots = f_n = 0$ inside $x\in\Csn$. Due to the Zariski closure in the definition of the resultants, the mixed $(A_0,\ldots,A_n)$-resultants may also describe solutions outside of $x\in\Csn$, e.g. ``roots at infinity''.\\
    
      If all polynomials $f_0,\ldots,f_n$ having the same monomial structure, i.e. $A_0 = \ldots = A_{n} =: A$, we will call $R_A(f_0,\ldots,f_{n}):=R_{A,A,\ldots,A}(f_0,\ldots,f_{n})$ simply the \textit{$A$-resultant}. The $A$-resultants satisfy a natural transformation law.
      
      \begin{lemma}[Transformation of $A$-resultants \cite{GelfandDiscriminantsResultantsMultidimensional1994}]
        Consider the polynomials $f_0,\ldots,f_n\in\mathbb C^A$ and let $D$ be an invertible $(n+1)\times (n+1)$ matrix. For the transformation $g_i = \sum_{j=0}^{n} D_{ij} f_j$ for $i=0,\ldots,n$ we have
        \begin{align}
          R_A(g_0,\ldots,g_n) = \det (D)^{\vol (P)} R_A (f_0,\ldots,f_n)
        \end{align}
        where $P = \Conv (A)$.
      \end{lemma}
      
      Especially, for linear functions $g_0,\ldots,g_n$ with $g_i:= \sum_{j=0}^{n} D_{ij}x_j$ in homogenization, that lemma implies $R_A(g_0,\ldots,g_n) = \det(D)$. This result extends also to all cases, where $A$ forms a simplex \cite{GelfandDiscriminantsResultantsMultidimensional1994}.
  
    \subsection{$A$-discriminants} \label{sec:Adisc}
      
      The $A$-discriminant is closely related to the $A$-resultant and describes for a given polynomial $f\in\mathbb C[x_1^{\pm 1},\ldots,x_n^{\pm 1}]$ when the hypersurface $\{f=0\}$ is singular. Equivalently, the $A$-discriminant determines whether $f$ has multiple roots. Let $A\subset \mathbb Z^n$ be the support of the polynomial $f(x) = \sum_{a\in A} z_a x^a$ and consider
      \begin{align}
        \nabla_0 = \left\{ f\in\mathbb C^A | \mathbf V \left(f,\frac{\partial f}{\partial x_1},\ldots,\frac{\partial f}{\partial x_n}\right) \neq\emptyset \text{ in } (\mathbb C^*)^n \right\} \subseteq \mathbb C^A \label{eq:Adisc}
      \end{align}
      the set of polynomials $f\in\mathbb C^A$ for which there exists a solution $x \in (\mathbb C^*)^n$ such that $f$ and its first derivatives vanish simultaneously. Analogue to the $A$-resultant, if the Zariski closure of $\nabla_0$ has codimension $1$, we will set the \textit{$A$-discriminant} $\Delta_A(f)\in\mathbb Z [\{z_a\}_{a\in A}]$ of $f$ as the defining polynomial of the hypersurface $\overline \nabla_0$. For higher codimensions $\operatorname{codim} (\overline{\nabla}_0) > 1$ we will fix the $A$-discriminant to be $1$. Configurations $A$, which having $\Delta_A(f)=1$ are called \textit{defective}. Combinatorial criteria of defective configurations can be found in \cite{EsterovNewtonPolyhedraDiscriminants2010, CurranRestrictionADiscriminantsDual2006, DickensteinTropicalDiscriminants2007}. By definition, the $A$-discriminant is an irreducible polynomial in the coefficients $z_a$, which is uniquely determined up to a sign \cite{GelfandDiscriminantsResultantsMultidimensional1994}. 
      \\
      
      \begin{example} \label{ex:cubicDisc}
        Consider the cubic polynomial in one variable $f = z_0 + z_1 x + z_2 x^2 + z_3 x^3$ with its support $A = \{0,1,2,3\}$. Its $A$-discriminant is given (up to a sign) by
        \begin{align}
          \Delta_A(f) = z_1^2 z_2^2 - 4 z_0 z_2^3 - 4 z_1^3 z_3 + 18 z_0 z_1 z_2 z_3 - 27 z_0^2 z_3^2
        \end{align}
        which can be calculated either by hand by eliminating $x$ from $f(x) = \pd{f(x)}{x} = 0$ or by a convenient mathematical software program e.g. Macaulay2 \cite{GraysonMacaulay2SoftwareSystem} with additional libraries \cite{StaglianoPackageComputationsClassical2018, StaglianoPackageComputationsSparse2020}. Thus, the equations $f(x) = \pd{f(x)}{x} = 0$ have a common solution for $x\neq 0$ if and only if $\Delta_A(f) = 0$.
      \end{example}
      
      However, many polynomials share the same $A$-discriminant. Consider two finite subsets $A\subset\mathbb Z^n$ and $A'\subset\mathbb Z^m$, which are related by an injective, affine transformation $T:\mathbb Z^n \rightarrow \mathbb Z^m$ with $T(A)=A'$. Then, the corresponding transformation of $T$ connects also $\Delta_A$ with $\Delta_{A'}$, which was shown in \cite{GelfandDiscriminantsResultantsMultidimensional1994}. Thus, the $A$-discriminant only depends on the affine geometry of $A$. For example the homogeneous polynomial $\tilde f\in\mathbb C [x_0,\ldots,x_n]$ with $\tilde f\in\mathbb C^{\Aa}$ can be dehomogenized by the map
      \begin{align}
        \mathbb C^{\Aa} \rightarrow \mathbb C^A, \qquad \tilde f(x_0,\ldots,x_n) \mapsto f(x_1,\ldots,x_n) = \tilde f(1,x_1,\ldots,x_n) \point
      \end{align}
      We can identify $\Delta_{\Aa}(\tilde f)=\Delta_A (f)$. Similarly we obtain for a finite subset $A\subset \mathbb Z^n$ and its homogenization $\Aa\subset\mathbb Z^{n+1}$ the same discriminants.\\
      
      By definition (\ref{eq:Adisc}) it can be seen, that $\Delta_A(f)$ has to be a homogeneous polynomial. Additionally, $\Delta_A(f)$ is quasi-homogeneous for any weight defined by a row of $A$ \cite{GelfandDiscriminantsResultantsMultidimensional1994}.  Removing these homogenities leads us to the \textit{reduced $A$-discriminant} $\Delta_B$. Let $\Aa\subset\mathbb Z^{n+1}$ be the homogenization of the support $A$ and $\mathcal B$ a Gale dual of $\Aa$. Then we can introduce ``effective'' variables
      \begin{align}
        y_j = \prod_{i=1}^N z_i^{b_{ij}} \qquad\text{for}\quad j=1,\ldots, r
      \end{align}
      where $b_{ij}$ denotes the elements of the Gale dual $\mathcal B$. Then, we can always rewrite the $A$-discriminant as $\Delta_A(f) = z^{\Lambda} \Delta_B(f)$, where the reduced $A$-discriminant $\Delta_B(f)$ is an inhomogeneous polynomial in the effective variables $y_1,\ldots,y_r$ and $\Lambda\in\mathbb Z^N$ defines a factor $z^\Lambda$. We will choose the smallest $\Lambda$ such that $\Delta_B$ is a polynomial.\\
      
      \begin{example}[Continuation of example \ref{ex:cubicDisc}] \label{ex:cubicGale}
        As a possible choice for the Gale dual of the homogenized $\Aa\subset\mathbb Z^2$ one obtains
        \begin{align}
          \mathcal B = \begin{pmatrix}
                         1 & 2 \\
                         -2 & -3\\
                         1 & 0\\
                         0 & 1
                       \end{pmatrix}\text{ ,} \qquad y_1= \frac{z_0 z_2}{z_1^2}, \qquad y_2 = \frac{z_0^2 z_3}{z_1^3} \point \label{eq:exEffective}
        \end{align}
        Thus, we can rewrite the $A$-discriminant 
        \begin{align}
          \Delta_A(f) = \frac{z_1^6}{z_0^2} \left(27 y_2^2 + 4 y_1^3 + 4 y_2 - y_1^2 - 18 y_1y_2\right) =  \frac{z_1^6}{z_0^2} \Delta_B(f)
        \end{align}
        as a reduced discriminant $\Delta_B(f)\in\mathbb Z[y_1,y_2]$.
      \end{example}
      
      Except for special cases, where $A$ forms a simplex or a circuit \cite{GelfandDiscriminantsResultantsMultidimensional1994}, the determination of the $A$-discriminant can be an extremely hard problem for bigger polynomials and calculations quickly get out of hand. Fortunately, there is an indirect description of $A$-discriminants which was given by Kapranov \cite{KapranovCharacterizationAdiscriminantalHypersurfaces1991} with slightly adjustments in \cite{CuetoResultsInhomogeneousDiscriminants2006}. This so called \textit{\HKP} states a very efficient way to study discriminants. Let $\mathcal S\subset(\mathbb C^*)^r$ be the hypersurface defined by $\Delta_B(f)=0$. Then this hypersurface $\mathcal S$ can be parametrized by $\psi :\mathbb P^{r-1}_{\mathbb C} \rightarrow (\mathbb C^*)^r$, where $\psi$ is given by
      \begin{align}
        \psi [t_1 : \ldots : t_r ] = \left( \prod_{i=1}^N \left(\sum_{j=1}^r b_{ij} t_j \right)^{b_{i1}}, \ldots, \prod_{i=1}^N \left(\sum_{j=1}^r  b_{ij} t_j \right)^{b_{ir}}\right)
      \end{align}
      where $b_{ij}$ are again the elements of a Gale dual $\mathcal B$ of $\Aa$. Hence, we can give an implicit representation of the $A$-discriminant very quickly only by knowing a Gale dual.
    
      \begin{example}[Continuation of example \ref{ex:cubicGale}] \label{ex:cubicParametrization}
        For the example of the cubic polynomial in one variable, we obtain
        \begin{align}
          \psi [t_1 : t_2] = \left( \frac{t_1 + 2 t_2}{(2t_1+3t_2)^2} t_1 , - \frac{(t_1 + 2 t_2)^2}{(2t_1+3t_2)^3} t_2\right) \point
        \end{align}
        Since $[t_1 : t_2]$ are homogeneous coordinates, only defined up to multiplication, we can set w.l.o.g. $t_2=1$. Hence, the statement of \HKP is, that for every $t_1\in\mathbb C$ we can identify
        \begin{align}
          y_1 = \frac{t_1 + 2 }{(2t_1+3)^2} t_1, \qquad y_2 = - \frac{(t_1 + 2)^2}{(2t_1+3)^3}
        \end{align}
        as the points characterizing the hypersurface $\Delta_B(f)(y_1,y_2) = 0$ or equivalently the hypersurface $\Delta_A(f)(z_0,z_1,z_2,z_3)=0$ by means of the relations (\ref{eq:exEffective}).
      \end{example}
    
      Moreover, Kapranov showed \cite{KapranovCharacterizationAdiscriminantalHypersurfaces1991}, that the map $\psi$ is the inverse of the (logarithmic) Gauss map, which is defined for an arbitrary hypersurface $\mathcal S_g = \{y\in(\mathbb C^*)^r | g(y)=0\}$ as $\gamma : (\mathbb C^*)^r \rightarrow \mathbb P^{r-1}_{\mathbb C}$, with $\gamma(y) = [ y_1 \partial_1 g(y) : \ldots : y_r \partial_r g(y)]$ for all regular points of $\mathcal S_g$. It is a remarkable fact, that all hypersurfaces $S_g$, which have a birational Gauss map are precisely those hypersurfaces defined by reduced $A$-discriminants \cite{KapranovCharacterizationAdiscriminantalHypersurfaces1991, CuetoResultsInhomogeneousDiscriminants2006}. \\
    
      To conclude this section, we want to mention the relation between $A$-discriminants and the mixed $(A_0,\ldots,A_n)$-resultants, which is also known as Cayley's trick.
      \begin{lemma}[Cayley's trick \cite{GelfandDiscriminantsResultantsMultidimensional1994}]
        Let $A_0,\ldots,A_{n}\subset \mathbb Z^n$ be finite subsets jointly generating $\mathbb Z^n$ as an affine lattice. By $f_i\in\mathbb C^{A_i}$ we denote the corresponding polynomials of the sets $A_i$. Then we have
        \begin{align}
          R_{A_0,\ldots,A_{n}} (f_0,\ldots,f_{n}) = \Delta_A \left(f_{0}(x) + \sum_{i=1}^n y_i f_i(x)\right)
        \end{align}
        where $A$ is the support of the polynomial $f_{0}(x) + \sum_{i=1}^n y_i f_i(x)\in \mathbb C[x_1^{\pm 1},\ldots,x_{n}^{\pm 1},y_1,\ldots,y_n]$.        
      \end{lemma}
    
      Thus, we can interpret the mixed $(A_0,\ldots,A_n)$-resultants as a special case of $A$-discriminants.

    \subsection{principal $A$-determinants}
    
      The last object we want to introduce from the book of Gelfand, Kapranov and Zelevinsky \cite{GelfandDiscriminantsResultantsMultidimensional1994} is the principal $A$-determinant. However, before introducing the principal $A$-determinant, which will be the main object for the following approach, we want to recall first some constructions of polytopes. 
      
      Let $A\subset \mathbb Z^n$ be a finite subset of points with cardinality $N$ and $P=\Conv(A)$ their convex hull. A \textit{triangulation} $T$ of a polytope $(P,A)$ is a set of simplices $\Conv(\sigma)$ with all $\sigma\subseteq A$, such that the union of all these simplices is equal to the full polytope $P=\cup_{\sigma\in T} \Conv(\sigma)$ and the intersection of two simplices $\Conv(\sigma_i)\cap\Conv(\sigma_j)$ is either empty or a face of both simplices \cite{DeLoeraTriangulations2010}. For every triangulation $T$ of a polytope $(P,A)$ we can introduce the weight map $\omega_T:A\rightarrow \mathbb Z_{\geq 0}$
      \begin{align}
        \omega_T (a) := \sum_{\substack{\sigma\in T \,\,\text{s.t.}\\ a\in\operatorname{Vert}(\Conv(\sigma))}} \vol (\Conv(\sigma))
      \end{align}
      which is the sum of all simplex volumes, having $a$ as its vertex and we write $\omega_T(A) = (\omega_T(a^{(1)}),\ldots,\omega_T(a^{(N)}))$ for the image of $A$.
      
      The weights themselves define a further polytope of dimension $r:=N-n-1$, which is the so-called \textit{secondary polytope} $\Sigma(\mathrm A)$. It is the convex hull of all weights
      \begin{align}
        \Sigma (A) := \Conv\left(\omega_T(A) | T \text{ is a triangulation of } A\right) \subset\mathbb R^N \point
      \end{align}
      
      A triangulation $T$, which correspond to a vertex in the secondary polytope, is called a \textit{regular triangulation}.\\
      
      Let us now define the principal $A$-determinant, which is a special $A$-resultant. Once again, we consider a finite subset $A\subset \mathbb Z^n$ and we will assume for the sake of simplicity that $\Aff_{\mathbb Z} (A) = \mathbb Z^n$. By $f=\sum_{a\in A} z_a x^a\in\mathbb C^A$ we denote the corresponding polynomial to $A$. The \textit{principal $A$-determinant} is then defined as
      \begin{align}
        E_A(f) := R_A\left(f,x_1\frac{\partial f}{\partial x_1},\ldots,x_n \frac{\partial f}{\partial x_n}\right) \point \label{eq:principalAdetResultant}
      \end{align}
      Thus, the principal $A$-resultant is a polynomial with integer coefficients depending on $\{z_a\}_{a\in A}$ and is uniquely determined up to a sign \cite{GelfandDiscriminantsResultantsMultidimensional1994}.\\
      
      The principal $A$-determinant can be decomposed into a product of several $A$-discriminants:
      \begin{theorem}[Prime factorization of principal $A$-determinant \cite{GelfandDiscriminantsResultantsMultidimensional1994}] \label{thm:pAdet-factorization}
        The principal $A$-determinant can be written as a product of $A$-discriminants
        \begin{align}
          E_A(f) = \pm \prod_{\tau \subseteq \operatorname{Newt} (f)} \Delta_{A\cap\tau} (f_\tau)^{\mu(A,\tau)}
        \end{align}
        where the product is over all faces $\tau$ of the Newton polytope $\Newt (f)$ and $ \mu(A,\tau) \in\mathbb N_{> 0}$ are certain integers, called multiplicity of $A$ along $\tau$. For the following the exact definition of the multiplicities is not crucial, which is why we refer to \cite{GelfandDiscriminantsResultantsMultidimensional1994, ForsgardZariskiTheoremMonodromy2020} at this point.
      \end{theorem}
      
      Most often we are only interested in the roots of the principal $A$-determinant. Therefore, we want to define a \textit{simple principal $A$-determinant} according to \cite{ForsgardZariskiTheoremMonodromy2020} where all multiplicities are set to $1$
      \begin{align}
        \widehat E_A(f) = \pm \prod_{\tau \subseteq \operatorname{Newt} (f)} \Delta_{A\cap\tau} (f_\tau)
      \end{align}
      which generates the same variety as $E_A(f)$. \\

      \begin{example}[principal $A$-determinants of homogeneous polynomials] 
        To illustrate the principal $A$-determinant we will recall an example from \cite{GelfandDiscriminantsResultantsMultidimensional1994}. Let $\tilde f\in\mathbb C[x_0,\ldots,x_n]$ be a homogeneous polynomial consisting in all monomials of a given degree $d\geq 1$. The support of $\tilde f$ will be called $\Aa$ and its Newton polytope $\Newt(\tilde f)\subset \mathbb R^{n+1}$ is an $n$-dimensional simplex, having $2^{n+1}-1$ faces. Moreover, the faces $\tau\subseteq\Newt(\tilde f)$ are generated by all non-empty subsets of the vertices $\{0,\ldots,n\}$ and one can show that all multiplicities $\mu(\Aa,\tau)$ are equal to one \cite{GelfandDiscriminantsResultantsMultidimensional1994}. Thus, we get
        \begin{align}
          E_\Aa(\tilde f) = \pm \prod_{\emptyset\neq \tau \subseteq \{0,\ldots,n\}} \Delta_{\Aa\cap \tau} (\tilde f_{\tau}) \label{eq:pAdetHom} \point
        \end{align}
        Note, that $\tilde f_{\tau} (x) = \tilde f(x)|_{x_i=0, i\notin \tau}$ is nothing else, than the polynomial $\tilde f$, where all variables $x_i=0$ set to be zero, which are corresponding to elements in the complement of $\tau$. That decomposition is exactly the behaviour we would naively expect in the polynomial equation system
        \begin{align}
          \tilde f (x) = x_0 \frac{\partial \tilde f(x)}{\partial x_0} = \ldots =  x_n \frac{\partial \tilde f(x)}{\partial x_n} = 0 \point
        \end{align}
        It is, that we can consider all combinations of several $x_i=0$ separately, except for the case, where all $x_i=0$ vanish. However, it should be remarked, that this behaviour is not true in general and we have rather to take the truncated polynomials into account as described in theorem \ref{thm:pAdet-factorization}. \\
      \end{example}

      Another noteworthy result of Gelfand, Kapranov and Zelevinsky is the connection between the principal $A$-determinant and the triangulations of $\Conv(A)$.
      \begin{theorem}[\cite{GelfandNewtonPolytopesClassical1990, SturmfelsNewtonPolytopeResultant1994, GelfandDiscriminantsResultantsMultidimensional1994}] \label{thm:NewtSec}
        The Newton polytope of the principal $A$-determinant and the secondary polytope coincide
        \begin{align}
          \Newt(E_A(f)) = \Sigma(A) \point
        \end{align}
        Further, if $T$ is a regular triangulation of $(\Conv(A),A)$, then the coefficient of the monomial $\prod_{a\in A} z_a^{\omega_T(a)}$ in $E_A(f)$ is equal to
        \begin{align}
          \pm \prod_{\sigma\in T} \vol(\sigma)^{\vol(\sigma)} \point
        \end{align}
        And also the relative signs of the coefficients in $E_A(f)$ can be fixed \cite[chapter 10.1G]{GelfandDiscriminantsResultantsMultidimensional1994}.
      \end{theorem}
      
      Thus, by the knowledge of all regular triangulations we can approximate the form of the principal $A$-determinant. As theorem \ref{thm:NewtSec} gives us the extreme monomials we can make a suitable ansatz for the principal $A$-determinant. The unknown coefficients of the monomials corresponding to potential interior points of the Newton polytope $\Newt(E_A(f))$ can be determined then by \HKP.
       
      Further, the theorem explains the typical appearing coefficients\footnote{Based on that theorem, Gelfand, Kapranov and Zelevinsky \cite{GelfandDiscriminantsResultantsMultidimensional1994} speculate on a connection between discriminants and probabilistic theory. Their starting point for this consideration are the ``entropy-like'' formulas as $\prod v_i^{v_i}=e^{\sum v_i \log v_i}$ for the coefficients of the principal $A$-determinant, as well as in the \HKP. In the latter case we can define an ``entropy'' $S := \ln \left (\prod_{l=1}^r \psi_l^{t_l}\right) = \sum_{i=1}^N \rho_i(t) \ln(\rho_i(t))$ where $\rho_i(t):= \sum_{j=1}^r b_{ij} t_j$. To the author's knowledge, more rigorous results about such a potential relation are missing. However, there are further connections known between tropical toric geometry and statistical thermodynamics as presented in \cite{KapranovThermodynamicsMomentMap2011, PassareAmoebasComplexHypersurfaces2013}. In any case, a fundamental understanding of such a relation could be very inspiring for a physical point of view.} in principal $A$-determinants and Landau varieties, like $1=1^1$, $4=2^2$, $27=3^3$, $4^4=256$.

      \begin{example}
        We will continue the example from section \ref{sec:Adisc}. From the set $A=(0,1,2,3)$ we can determine $4$ triangulations with weights $\omega_1 = (3,0,0,3)$, $\omega_2=(1,3,0,2)$, $\omega_3=(2,0,3,1)$ and $\omega_4=(1,2,2,1)$. Adding the only possible interior point $(2,1,1,2)$, we obtain the following ansatz
        \begin{align}
          E_A(f) = 27 z_0^3 z_3^3 + 4 z_0 z_1^3 z_3^2 + 4 z_0^2 z_2^3 z_3 - z_0 z_1^2 z_3 + r z_0^2 z_1 z_2 z_3^2
        \end{align}
        where we have to determine $r\in\mathbb Z$. By considering the faces of $\Newt(f)$ we can split $E_A(f)$ into discriminants
        \begin{align}
          E_A(f) = z_0 z_3 \Delta_A(f) = z_0 z_3 \frac{z_1^6}{z_0^2} \Delta_B(f)
        \end{align}
        with the reduced $A$-discriminant $\Delta_B(f) = 27 y_2^2 + 4y_1^3 +4y_2 -y_1^2 + r y_1 y_2$ with the same conventions as in example \ref{ex:cubicGale}.  By the \HKP (see example \ref{ex:cubicParametrization}) we can calculate those points $(y_1,y_2)$ which satisfy $\Delta_B(f)=0$. Choosing for example $t_1=-1$, we obtain the point $(y_1,y_2)=(-1,-1)$, which leads to $r=-18$.
      \end{example}
      
      In general we can replace $E_A(f)$ by means of Cayley's trick by one single $A$-discriminant in order to simplify the usage of Horn-Kapranov-parametrization. However, we have to mention that the number of vertices of the secondary polytope -- or equivalently the number of regular triangulations -- grows very fast. In the application of Feynman integrals (see below) the $2$-point, $3$-loop Feynman diagram (also known as $3$-loop banana) has $79.884$ possible triangulations. The $2$-loop double-edged triangle graph (or dunce's cap) generates even $889.044$ triangulations. Nevertheless, this approach could be faster, than the direct calculation of principal $A$-determinants by standard algorithms, since there are very efficient methods for triangulations known \cite{DeLoeraTriangulations2010}.
 
%
%---------------------------------------------------------------------  
%

  \section{$\Aa$-hypergeometric systems} \label{sec:AHyp}
 
    In this section we will establish the link between $\Aa$-hypergeometric systems and the $A$-resultants. This connection was developed in a series of articles by Gelfand, Kapranov and Zelevinsky, mainly in \cite{GelfandEquationsHypergeometricType1988, GelfandAdiscriminantsCayleyKoszulComplexes1990, GelfandHypergeometricFunctionsToric1991, GelfandDiscriminantsResultantsMultidimensional1994}. A major part of this correspondence was also shown in \cite{AdolphsonHypergeometricFunctionsRings1994}, where the following deduction is mostly based on. A generalization of this relation can be found in \cite{BerkeschZamaereSingularitiesHolonomicityBinomial2014} and \cite{SchulzeIrregularityHypergeometricSystems2006}. \\
    
    Let $\Aa=\{a^{(1)},\ldots,a^{(N)}\}\subset \mathbb Z^{n+1}$ be a finite subset of lattice points, spanning $\mathbb R^{n+1}$ as a vector space $\operatorname{span}_{\mathbb R}(\Aa) = \mathbb R^{n+1}$. Equivalently, we can demand $\Aa$ to have full rank. We will usually consider the case $n+1\leq N$. Further, we will assume that there exists a linear map $h:\mathbb Z^{n+1}\rightarrow \mathbb Z$, such that $h(a)=1$ for any $a\in\Aa$. The latter means, that all elements of $\Aa$ lie on a common hyperplane off the origin, which will allow us to consider $\Aa$ as elements of an affine space in homogenization. Alternatively, we can also demand $f=\sum_{a\in\Aa} z_a x^a$ to be quasi-homogeneous. \\
    
    The $\Aa$-hypergeometric system is a left ideal in the Weyl algebra $\mathcal D_N := \langle z_1,\ldots,z_N,\allowbreak\partial_1,\ldots,\partial_N\rangle$ and will be generated by two types of differential operators
    \begin{align}
      \square_l &= \prod_{l_j>0} \partial_j^{l_j} - \prod_{l_j<0} \partial_j^{-l_j} \qquad\text{for}\quad l\in\mathbb L \label{eq:toricOperators}\\
      E_i(\beta) &= \sum_{j=1}^N a_i^{(j)} z_j \partial_j + \beta_i \qquad\text{for}\quad i=0,\ldots,n \label{eq:homogeneousOperators}
    \end{align}
    where $\mathbb L = \ker_{\mathbb Z}(\Aa)$ is the integer kernel of $\Aa$ and $\beta\in\mathbb C^{n+1}$ is an arbitrary complex number. Thus, the full set of differential operators is given by
    \begin{align}
      H_\Aa(\beta) = \sum_{i=0}^{n} \mathcal D_N E_i(\beta) + \sum_{l\in\mathbb L} \mathcal D_N \square_l \point
    \end{align}
    Holomorphic solutions on convenient domains in $\mathbb C^N$ of these differential equation systems will be called \textit{$\Aa$-hypergeometric functions} and we denote the solution space of those functions by
    \begin{align}
      \operatorname{Sol}(H_\Aa(\beta)) = \{ F\in\mathcal O | P \bullet F = 0 \quad \forall\, P\in H_\Aa(\beta)\}
    \end{align}
    where $\mathcal O$ is the $\mathcal D_N$-module of holomorphic functions on a convenient domain in $\mathbb C^N$. We refer to the $\mathcal D_N$-module of equivalence classes $\mathcal M_\Aa(\beta) = \mathcal D_N / H_\Aa(\beta)$ as the \textit{$\Aa$-hypergeometric module} and we have the isomorphism $ \operatorname{Sol}(H_\Aa(\beta)) \cong \operatorname{Hom}_{\mathcal D_N} \left(\mathcal M_\Aa(\beta), \mathcal O\right)$ \cite{BjorkAnalyticDModulesApplications1993}.\\
    
    We should next like to turn to the existence of such solutions as well as the singularities of their analytic continuation to the whole complex domain $\mathbb C^N$. We will mostly follow \cite{OakuComputationCharacteristicVariety1994} and \cite{SaitoGrobnerDeformationsHypergeometric2000}. Let us first recall the situation in the well-known one-dimensional case. Let $P=c_m(z) \partial^m + \ldots +c_1(z) \partial + c_0(z)$ be a differential operator in a single variable $z\in\mathbb C$ with polynomials $c_i(z)\in\mathbb C[z]$ as coefficients and $c_m(z) \not\equiv 0$. We will call the roots of the leading coefficient $c_m(z)$ the singular points of $P$ and the set of all these roots will be called the singular locus $\Sing(\mathcal D_1 P)$. Standard existence theorems state then, that for simply connected domains $U\subseteq \mathbb C\setminus \Sing(\mathcal D_1 P)$ outside of the singular locus, there are always holomorphic solutions $F$ of the ordinary linear differential equation $P \bullet F =0$ and the dimension of the solution space is equal to $m$.
    
    In the multivariate case we will consider differential operators $P=\sum_{u,v\in\mathbb N^N} c_{uv} z^u \partial^v \in \mathcal D_N$ with the order $\nu(P):=\max\{|v| : c_{uv}\neq 0\}$ instead of ordinary linear differential equations. For the differential operators $P$ we will define their \textit{principal symbols} as
    \begin{align}
      \sigma(P) = \sum_{u\in\mathbb N^N,v=\nu(P)} c_{uv} z^u \xi^v \in \mathbb C [z_1,\ldots,z_N,\xi_1,\ldots,\xi_N] \point
    \end{align}
    It is a polynomial in the commuting variables $z,\xi$ in the so-called associated graded ring of $\mathcal D_N$. Principal symbols are special cases of initial forms. For every proper left ideal $I\subset\mathcal D_N$ we will call $\{ \sigma(P) | P\in I\}$ the characteristic ideal and by 
    \begin{align}
      \ch(I) = \mathbf V \big(\sigma(I)\big) = \left\{(z,\xi)\in\mathbb C^{2N} \,\rvert\, \sigma(P)(z,\xi) = 0 \quad\text{for all}\quad P\in I\right\}
    \end{align}
    we denote the \textit{characteristic variety} of the ideal $I$, which is equal to the characteristic variety $\ch(\mathcal D_N / \mathcal D_N I)$ due to \cite{OakuComputationCharacteristicVariety1994}. We will call $I$ \textit{holonomic} if its characteristic variety $\ch(I)$ has the minimal Krull dimension $N$. In the multivariate case we define the \textit{singular locus} of the $\mathcal D$-module $\mathcal D_N / \mathcal D_N I$ to be the Zariski-closed projection of the characteristic variety $\ch(\mathcal D_N / \mathcal D_N I)\subset\mathbb C^{2N}$ without the trivial solution $\xi_1=\ldots=\xi_N=0$ to the $z$-space $\mathbb C^N$, i.e.
    \begin{align}
      \operatorname{Sing}(I) = \overline{\left\{ \hat z\in\mathbb C^N | (\hat z,\hat \xi)\in\ch(\mathcal D_N / \mathcal D_N I) \setminus \{\hat\xi_1=\ldots=\hat\xi_N=0\} \right\}} \point
    \end{align}
    If $I$ is a holonomic ideal, the singular locus is always a proper subset of $\mathbb C^N$. The existence and uniqueness of solutions of systems of partial differential equations is guaranteed by the Cauchy-Kovalevskaya-Kashiwara theorem, which we will recall in the version of \cite{SaitoGrobnerDeformationsHypergeometric2000}.
    \begin{theorem}[Cauchy-Kovalevskaya-Kashiwara theorem \cite{SaitoGrobnerDeformationsHypergeometric2000}] \label{thm:CKKT}
      Let $I$ be a holonomic $\mathcal D$-ideal and $U\subseteq \mathbb C^N\setminus\Sing(I)$ be a simply connected domain. The dimension of the solution space of the equation systems $P\bullet F = 0$ for all $P\in I$ with holomorphic functions $F$ on $U$ is always finite and equal to the holonomic rank of $I$.
    \end{theorem}
    
    Especially, for $\Aa$-hypergeometric systems we can simply determine the holonomic rank.
    \begin{theorem}[\cite{GelfandHypergeometricFunctionsToric1991, GelfandGeneralizedEulerIntegrals1990, CattaniThreeLecturesHypergeometric2006}]
      Let $H_\Aa(\beta)$ be an $\Aa$-hypergeometric system. $H_\Aa(\beta)$ is always holonomic and for generic $\beta\in\mathbb C^{n+1}$ we have
      \begin{align}
        \rank (H_\Aa(\beta)) = \vol (\Conv (\Aa)) \point
      \end{align}
    \end{theorem}
    
    Thus, the existence and uniqueness of $\Aa$-hypergeometric functions is guaranteed. In the remaining section we want to analyze the structure of the singular locus further. Let us first remark, that we can restrict us to the codimension $1$ part of the singular locus, since all singularities in higher codimensions are removable singularities due to Riemann's second removable theorem \cite{KaupHolomorphicFunctionsSeveral1983, BerkeschZamaereSingularitiesHolonomicityBinomial2014}. \\

    In order to connect the geometry of $\Conv(\Aa)$ with the structure of $\Sing(H_\Aa(\beta))$ we will state the following lemma.
    \begin{lemma} \label{lem:face-kernel}
      Let $l\in\mathbb L=\ker_{\mathbb Z}(\Aa)$ be an arbitrary element of the kernel of $\Aa$ and $\tau\subsetneq\Conv(\Aa)$ be an arbitrary face. Then $\{l_j\}_{j\notin\tau}$ consists either only in zeros or it contains both, elements with positive and with negative integer.
    \end{lemma}
    \begin{proof}
      Let $\phi:\mathbb R^{n+1}\rightarrow\mathbb R$ be the linear map which characterizes the face $\tau$, i.e. the linear function which is maximized exactly for all points in $\tau$. Thus, we have
      \begin{align}
        0 = \phi\left(\sum_{j=1}^N l_ja^{(j)}\right)= r \sum_{j\in \tau} l_j + \sum_{j\notin\tau} l_j \phi(a^{(j)}) = \sum_{j\notin\tau} l_j \left[ \phi(a^{(j)}) - r \right]
      \end{align}
      where we denote by $r=\max_{j}\left(\phi(a^{(j)})\right)$ the value, which $\phi$ reaches for all $j\in\tau$. Moreover, we make use of the homogenity of $\Aa$ which implies $\sum_{j}l_j = \sum_{j\in\tau} l_j + \sum_{j\notin\tau} l_j= 0$. On the other hand we have $\phi(a^{(j)})-r<0$ for all $j\notin\tau$. This shows the assertion.
    \end{proof}
    
    Furthermore, we can establish a connection between the faces of $\Conv(\Aa)$ and the characteristic variety of $\mathcal M_\Aa(\beta)$. The two following lemmata are inspired by \cite{AdolphsonHypergeometricFunctionsRings1994} with some slightly adjustments.
    
    \begin{lemma} \label{lem:face-characteristic}
      For every point $(\hat z,\hat \xi)\in\ch(\mathcal M_\Aa(\beta))$, there exists a unique face $\tau\subseteq\Conv(\Aa)$ such that $\hat \xi_j\neq 0$ if and only if $j\in\tau$.
    \end{lemma}
    \begin{proof}
     The case $\hat\xi = (0,\ldots,0)$ is trivially satisfied by $\tau=\emptyset$ and we will exclude this case in the following. Denote by $\emptyset\neq J\subseteq\{1,\ldots,N\}$ the index set for which $\hat\xi_j \neq 0$ for all $j\in J$.  Let $\tau$ be the carrier of $J$, i.e. the smallest face of $\Conv(\Aa)$ containing the points with labels in $J$. We want to show first, that $J$ spans affinely the supporting hyperplane of $\tau$, i.e. that $\Conv(J)$ and $\tau$ having the same dimension.
%      \footnote{If $\Conv(J)$ and $\tau$ having the same dimension and $J\subseteq \tau$, they also share the same affine span. Proof: From $J\subseteq\tau$ follows $\operatorname{Aff}(J)\subseteq\operatorname{Aff}(\tau)$ (see e.g. Philip. Linear Algebra I 2019, Prop. 5.9 ). Since $\operatorname{Aff}(J)$ has dimension $k$, we have a basis of $k+1$ affinely independent points. An arbitrary point of $\operatorname{Aff}(\tau)$ is either an affine combination of that basis or not. If it is not an affine combination we would have $k+2$ affinely independent points. That contradicts the dimension of $\operatorname{Aff}(\tau)$.}
      
      Suppose that $\dim(\tau) > \dim (\Conv(J))$. Then we can find two points $\alpha,\beta\in\tau\setminus J$ with $\hat\xi_\alpha = \hat\xi_\beta=0$, such that the line segment from $\alpha$ to $\beta$ has an intersection point with $\Conv(J)$. Thus, there exist a rational number $0<\gamma<1$ and rational numbers $\lambda_j\geq 0$ describing this intersection point
      \begin{align}
        \gamma a^{(\alpha)} + (1-\gamma) a^{(\beta)} = \sum_{j\in J} \lambda_j a^{(j)}
      \end{align}
      with $\sum_{j\in J} \lambda_j = 1$. Denote by $m\in\mathbb Z_{>0}$ the least common multiple of all denominators of $\gamma$ and $\lambda_j$ for $j\in J$. Then we can generate an element in $\mathbb L$ or in $H_\Aa(\beta)$, respectively
      \begin{align}
        \square = \partial_\alpha^{m\gamma} \partial_\beta^{m(1-\gamma)} - \prod_{j\in J} \partial_j^{m\lambda_j} \in H_\Aa(\beta) \point
      \end{align}
     Since its principal symbol
      \begin{align}
        \sigma(\square) = \hat\xi_\alpha^{m\gamma} \hat\xi_\beta^{m(1-\gamma)} - \prod_{j\in J} \hat\xi_j^{m\lambda_j}
      \end{align}
      has to vanish for all values $(\hat z,\hat \xi)\in\ch(\mathcal M_\Aa(\beta))$ we get a contradiction, since $\hat\xi_\alpha=\hat\xi_\beta=0$ and $\hat\xi_j\neq 0$ for all $j\in J$.
      
      Thus, $\Conv(J)$ and $\tau$ having the same dimension. The second step will be to show, that $J=\tau$. Let $k\in\tau$ be an arbitrary point of the face $\tau$. We then have to prove that $\hat\xi_k\neq 0$. Since $\tau$ lies in the affine span of $J$, we will find some rational numbers $\lambda_j$ such that
      \begin{align}
        a^{(k)} = \sum_{j\in J} \lambda_j a^{(j)} = \sum_{\substack{j\in J \\ \lambda_j<0}} \lambda_j a^{(j)} + \sum_{\substack{j\in J \\ \lambda_j\geq 0}} \lambda_j a^{(j)}
      \end{align}      
      with $\sum_{j\in J} \lambda_j = 1$. Again, let $m\in\mathbb Z_{>0}$ the least common multiple of all $\lambda_j$ with $j\in J$, which will generate an element in $\mathbb L$ and we obtain
      \begin{align}
        \square = \partial_{k}^m \prod_{\substack{j\in J \\ \lambda_j<0}} \partial_j^{-m \lambda_j} - \prod_{\substack{j\in J \\ \lambda_j\geq 0}} \partial_j^{m\lambda_j} \in H_\Aa(\beta) \point
      \end{align}
      Both terms having the same order, since $1-\sum_{\lambda_j<0} \lambda_j = \sum_{\lambda_j\geq 0} \lambda_j$, which results in
      \begin{align}
        \sigma(\square) = \xi_{k}^m \prod_{\substack{j\in J \\ \lambda_j<0}} \xi_j^{-m \lambda_j} - \prod_{\substack{j\in J \\ \lambda_j\geq 0}} \xi_j^{m\lambda_j} \point
      \end{align}
      Thus, it is $\hat\xi_{k}^m \prod_{\lambda_j<0} \hat\xi_j^{-m \lambda_j} = \prod_{\lambda_j\geq 0} \hat\xi_j^{m\lambda_j}$. By assumption it is $\hat\xi_{j}\neq 0$ for all $j\in J$ and therefore it follows also $\hat\xi_{k}\neq 0$. 
    \end{proof}
    
    In order to give a relation between $A$-discriminants and the characteristic varieties, we will associate to every finite subset $\Aa\subset\mathbb Z^{n+1}$ a multivariate polynomial 
    \begin{align}
      f_z(x) = \sum_{a^{(j)}\in\Aa} z_j x^{a^{(j)}} \in \mathbb C[x_0,\ldots,x_n] \point
    \end{align}
    Recall, that for every face $\tau\subseteq\Aa$ we understand by 
    \begin{align}
      f_{\tau,z}(x) = \sum_{j\in\tau} z_j x^{a^{(j)}} \in \mathbb C[x_0,\ldots,x_n]
    \end{align}
    the truncated polynomial with respect to the face $\tau$.
    
    \begin{lemma} \label{lem:characteristic-disc}
      Let $\emptyset\neq\tau\in\Conv(\Aa)$ be an arbitrary face. Then the following two statements are equivalent:
      \begin{enumerate}[i)]
        \item the point $(\hat z,\hat \xi)\in\ch(\mathcal M_\Aa(\beta))$ is a point of the characteristic variety and $\tau$ is the face corresponding to this point according to lemma \ref{lem:face-characteristic}, i.e. $\hat \xi_j \neq 0$ if and only if $j\in\tau$
        \item the polynomials $\frac{\partial f_{\tau,\hat z}}{\partial x_0},\ldots,\frac{\partial f_{\tau,\hat z}}{\partial x_n}$ have a common zero in $x\in (\mathbb C^\star)^{n+1}$
      \end{enumerate}        
    \end{lemma}
    \begin{proof}
      ``$ii)\Rightarrow i)$'': Let $\hat x\in (\mathbb C^\star)^{n+1}$ be a common solution of $\frac{\partial f_{\tau,\hat z}}{\partial x_0}=\ldots=\frac{\partial f_{\tau,\hat z}}{\partial x_n}=0$ which implies
      \begin{align}
        \hat x_i \frac{f_{\tau,\hat z}(\hat x)}{\partial x_i} = \sum_{j\in \tau} a_i^{(j)} \hat z_j \hat x^{a^{(j)}} = 0 \point
      \end{align}
      Consider the principal symbol of the homogeneous operators $E_i(\beta)\in H_\Aa(\beta)$ from equation (\ref{eq:homogeneousOperators}). By setting all $\hat\xi_j=0$ for $j\notin\tau$ and $\hat\xi_j = x^{a^{(j)}}$ for all $j\in\tau$ we obtain
      \begin{align}
        \sigma(E_i(\beta))(\hat z,\hat\xi) = \sum_{a^{(j)}\in\Aa} a_i^{(j)} \hat z_j \hat\xi_j = 0\point
      \end{align}
      It remains to prove that $\sigma(\square_l) (\hat z,\hat \xi) = 0$ for all $l\in\mathbb L$, where $\square_l$ was defined in equation (\ref{eq:toricOperators}). Since all points of $\Aa$ lying on a hyperplane off the origin, all monomials in $\square_l$ having the same order. Therefore, we have to show
      \begin{align}
        \prod_{l_j>0} \hat\xi_j^{l_j} = \prod_{l_j<0} \hat\xi_j^{-l_j} \quad\text{for all}\quad l\in\mathbb L \point \label{eq:xi=xi}
      \end{align}
      According to lemma \ref{lem:face-kernel} there are only two possible cases. In the first case we will have all $l_j=0$ with $j\notin\tau$. Thus, we insert $\hat\xi_j=\hat x^{a^{(j)}}$ for all $j\in\tau$
      \begin{align}
        \hat x^{\sum_{l_j>0} l_j a^{(j)}} = \hat x^{-\sum_{l_j<0} l_j a^{(j)}}
      \end{align}
      which is true, since all $l\in\mathbb L$ satisfy $\sum_j l_j a^{(j)} = \sum_{l_j>0} l_j a^{(j)} + \sum_{l_j<0} l_j a^{(j)} = 0$. In the second case, there are elements with $l_j<0$ as well as with $l_j>0$ corresponding to points outside of $\tau$ and equation (\ref{eq:xi=xi}) is trivially satisfied by $0=0$.\\
      
      ``$i)\Rightarrow ii)$'': If $(\hat z,\hat\xi)\in\ch(\mathcal M_\Aa(\beta))$ that states
      \begin{align}
        \sigma(E_i(\beta)) = \sum_{a^{(j)}\in\Aa} a^{(j)} \hat z_j \hat\xi_j = \sum_{j\in\tau} a^{(j)} \hat z_j \hat\xi_j = 0 \point
      \end{align}
      Thus, $\frac{\partial f_{\tau,\hat z}}{\partial x_0},\ldots,\frac{\partial f_{\tau,\hat z}}{\partial x_n}$ have a common zero in $\hat x\in(\mathbb C^*)^{n+1}$ if the system of equations
      \begin{align}
        \hat x^{a^{(j)}} = \hat \xi_j \qquad\text{for all}\quad j\in\tau \label{eq:xxi}
      \end{align}
      has a solution in $\hat x\in(\mathbb C^*)^{n+1}$. Hence, we have to show that it is impossible to construct a contradicting equation by combining the equations of (\ref{eq:xxi}). In other words for all integers $l_j\in \mathbb Z$ satisfying 
      \begin{align}
        \sum_{j\in\tau} l_j a^{(j)} = 0 \label{eq:lcomb}
      \end{align}
      we have to show that $\prod_{j\in\tau} (\hat\xi_j)^{l_j} = 1$. Note, that (\ref{eq:lcomb}) directly give rise to an element in $\mathbb L$, by setting the remaining $l_j=0$ for all $j\notin\tau$. Therefore, we can construct
      \begin{align}
        \square = \prod_{l_j>0} \partial_j^{l_j} - \prod_{l_j<0} \partial_j^{-l_j} \in H_\Aa(\beta) \point \label{eq:squarel}
      \end{align}
      Again, by the fact that all points of $\Aa$ lie on a common hyperplane off the origin, both terms in (\ref{eq:squarel}) have the same order. Thus,
      \begin{align}
        \sigma(\square) = \prod_{l_j>0} \xi_j^{l_j} - \prod_{l_j<0} \xi_j^{-l_j} 
      \end{align}
      which completes the proof since $\sigma(\square)(\hat z,\hat \xi) = 0$.
    \end{proof}
   
    By the previous lemma, we can conclude directly:
    \begin{cor}
      Let $A\subset\mathbb Z^n$ be a finite subset, $\Aa\subset\mathbb Z^{n+1}$ its homogenization and $f=\sum_{a^{(j)}\in A} z_j x^{a^{(j)}}\in\mathbb C[x_1,\ldots,x_n]$ the corresponding polynomial. Then we have the equality
      \begin{align}
        \operatorname{Sing}(\mathcal M_\Aa(\beta)) = \mathbf V (E_A(f)) \point
      \end{align}
    \end{cor}
    \begin{proof}
      Suppose that $\Aa$ has the form $\{(1,a^{(1)}),\ldots,(1,a^{(N)})\}$ with $A=\{a^{(1)},\ldots,a^{(N)}\}$. Then the statement ii) in lemma \ref{lem:characteristic-disc} is equal to a common zero of $f_{\tau,\hat z}, \pd{f_{\tau,\hat z}}{x_1},\ldots,\pd{f_{\tau,\hat z}}{x_n}$ in $x\in \Csn$. Thus, the singular locus $\Sing(\mathcal M_\Aa(\beta))$ is given by the Zariski closure of
      \begin{align}
        \bigcup_{\emptyset\neq \tau\subseteq\Conv(A)} \left\{ \hat z \in\mathbb C^N | \mathbf V \left(f_{\tau,\hat z},\frac{\partial f_{\tau,\hat z}}{\partial x_1},\ldots,\frac{\partial f_{\tau,\hat z}}{\partial x_n}\right) \neq\emptyset \text{ in } (\mathbb C^*)^n \right\} \point
      \end{align}
      But this is nothing else than the union of all $A$-discriminants
      \begin{align}
        \Sing(\mathcal M_\Aa(\beta)) = \bigcup_{\emptyset\neq \tau\subseteq\Conv(A)} \mathbf V (\Delta_{A\cap\tau} (f_\tau)) \point
      \end{align}
      The application of theorem \ref{thm:pAdet-factorization} concludes the proof. Since the $A$-discriminants only depend on the affine structure of $A$, the corollary also applies to configurations $\Aa$ which lying on a common hyperplane off the origin, but being not necessarily in the form $\{(1,a^{(1)}),\ldots,(1,a^{(N)})\}$.
    \end{proof}
    
    Thus, we have characterized the singular locus of $\Aa$-hypergeometric modules, which describes the possible singularities of the $\Aa$-hypergeometric functions, by the principal $A$-determinant. In general it is a hard problem to calculate these principal $A$-determinants. However, by the \HKP we have a way to describe these possible singularities very efficiently in an indirect manner.

%------------------------------------------------------------------------    

  \section{Feynman integrals and Landau varieties}
    
    \subsection{Basic definitions}
      \label{sec:FeynmanIntegrals}
    
      After introducing all the methods and tools in the previous sections, let us now turn to the main object of interest in this work: the Feynman integral. Appearing in all perturbative quantum field theories, Feynman integrals are an indispensable building block for almost every prediction within these theories. In order to focus on the main issues, we will restrict our discussion to scalar Feynman integrals only. This restriction is vindicated by the fact, that in principle every other Feynman integral can be reduced to a linear combination of scalar Feynman integrals, e.g. by the techniques of Passarino and Veltman \cite{PassarinoOneloopCorrectionsAnnihilation1979}, Tarasov's method \cite{TarasovGeneralizedRecurrenceRelations1997, TarasovReductionFeynmanGraph1998} or by the Corolla polynomial \cite{KreimerPropertiesCorollaPolynomial2012}. However, it should not go unmentioned that the reduction can be very difficult in practice.
      
      Formally speaking, a Feynman integral maps a Feynman graph $\Gamma$ to a multivalued function $I_\Gamma$, which depends on several variables and parameters. As specific representations of these functions $I_\Gamma$, we can write down different kinds of integrals, each valid only on a restricted domain. Thus, we do not want to use the term ``Feynman integral'' to refer to individual integrals, but rather to the analytical, common continuation of these integrals to a maximal domain for the parameters and the variables. The integral representations of Feynman integrals can be classified by the integration variables into position space representation, momentum space representation and parametric space representations.
      
      In momentum space, we assign to every external edge (also called ``leg'') of the Feynman graph $\Gamma$ a momentum $p_1,\ldots,p_E$, which will be treated as a given variable. The internal edges $e_1,\ldots,e_n$ of the Feynman graph $\Gamma$ are assigned with momenta $q_i\in\mathbb R^d$, such that momentum conservation is satisfied at every vertex of $\Gamma$. The remaining momenta, which are not determined by this procedure will be denoted by $k_1,\ldots,k_L$, where $L$ is the number of loops, i.e. the first Betti number of $\Gamma$. The representation of the Feynman integral in momentum space is the integral over all these undetermined momenta $k_1,\ldots,k_L$
      \begin{align}
        I_\Gamma (d,\nu,p,m) = \int_{\mathbb R^{d\times L}} \prod_{j=1}^L \frac{d^d k_j}{\pi^{d/2}} \prod_{i=1}^n \frac{1}{D_i^{\nu_i}} \label{eq:FeynmanMomSp} \comma
      \end{align}
      where external momenta $p = (p_1,\ldots,p_E)^T\in\Omega_p\subseteq\mathbb C^{d\times E}$ and internal masses $m = (m_1,\ldots,m_n)^T\in\Omega_m\subseteq\mathbb C^n$ considered as variables, whereas the spacetime dimension $d$ and the indices $\nu=(\nu_1,\ldots,\nu_n)$ will be treated as parameters. The denominators of the integrand $D_i=q_i^2+m_i^2$ stand for the inverse propagators, which we attach to every edge $e_i$ of the graph $\Gamma$.
      
      Note, that all momenta $p_i,q_i,k_j$ are supposed to be $d$-dimensional vectors\footnote{To be precise the function $I_\Gamma$ only depends on scalar products $s_{ij}=p_ip_j$ of external momenta and we will proceed to consider $I_\Gamma$ as depending on scalar products $s_{ij}$ and squares of momenta $m_i^2$, which will both considered to be complex numbers. This will be more apparent in the parametric representations.}. By restricting the external momenta to be real numbers $p\in\mathbb R^{d\times E}$, the integral in (\ref{eq:FeynmanMomSp}) is often referred to be the ``Euclidean'' version of the Feynman integral, since all momenta are vectors in a real Euclidean space. However, in a physical context one is interested in a slightly different version of (\ref{eq:FeynmanMomSp}), where we replace $D_i$ by $\tilde D_i = -q_i^2 + m_i^2$ and suppose all momenta to be $d$-dimensional vectors of Minkowski space. Since all momenta only appear as squares, we can change from (\ref{eq:FeynmanMomSp}) to its ``Minkowskian'' version by considering one component of the momenta to be purely imaginary. This procedure is known in literature as Wick rotation. Thus, by considering the momenta $p$ to be complex-valued we can include both in (\ref{eq:FeynmanMomSp}): the Euclidean and the Minkowskian version.
      
      However, in the analytic continuation from real momenta to complex momenta, we will find a serious issue, which is the underlying reason for considering the approach of this article: for complex momenta $p$ we can not always ensure, that $D_i\neq 0$. This fact implies the kinematic singularities of Feynman integrals as well as the multivaluedness of Feynman integrals, since these singularities turn out to be branch points. In order to fix a principal sheet, it is a common practice to introduce a small imaginary part in the inverse propagators
      \begin{align}
        D_i = q_i^2 +m_i^2 - i\epsilon
      \end{align}
      with $\epsilon>0$, which we will sent to zero after integration. This so-called ``$i\epsilon$ prescription'' can be also comprised in a redefinition of the masses $ {m_i}^2  \mapsto m_i^2 -i\epsilon$. For simplicity we will drop this small imaginary part in the most notations. However, we will discuss this more in detail by means of the coamoeba in section \ref{sec:coamoebas}.    
      
      Also the indices $\nu\in\mathbb C^n$ will be continued to complex numbers, as well as the spacetime dimension $d$. However, we will not discuss the meaning of such that integrals as (\ref{eq:FeynmanMomSp}) with non-integer $d$ and we will refer to the parametric space integrals, which gives a precise meaning of complex-valued $d$. Equipped, with complex $d$ and $\nu$, Feynman integrals are also well prepared for regularizations, as dimensional and analytic regularization procedures. \\
      
      As stated above, the inverse propagators $D_i=q_i^2+m_i^2$ contain $q_i=\sum_{j=L}R_{ij}k_j+\sum_{j=1}^E S_{ij}p_j$ a linear combination of external momenta and loop momenta. Thus, we can sort the inverse propagators in terms being quadratic, linear or constant in the loop momenta $k=(k_1,\ldots,k_L)$
      \begin{align}
        \Lambda(k,p,x) := \sum_{i=1}^n x_i D_i =  k^T M k + 2 p^T B k + p^T C p + J \label{eq:LandauPsi}
      \end{align}
      with a symmetric $L\times L$ matrix $M=R^TXR$,  an $E\times L$ matrix $B=S^TXR$, a symmetric $E\times E$ matrix $C=S^TXS$ and $J=\sum_{i=1}^n x_i m_i^2$. We will call $x_1,\ldots,x_n$ the Schwinger parameters and we will collect them by $X=\diag(x_1,\ldots,x_n)$.
      
      In doing so, we can construct two graph polynomials in the Schwinger parameters: the first and the second Symanzik polynomial (see e.g. \cite{EdenAnalyticSmatrix1966, BognerFeynmanGraphPolynomials2010})
      \begin{align}
        \Uu  = \det (M)\,, \qquad \Ff = -p^T B \Adj (M) B^T p + \det (M)(p^TCp+J)  \point \label{eq:SymanzikDef1}
      \end{align}
      
      Alternatively, we can construct the Symanzik polynomials directly from the graphs by considering the set of all spanning trees $\mathcal T_1$ and the set of all spanning two-forests $\mathcal T_2$ of the Feynman graph $\Gamma$
      \begin{align}
        \Uu = \sum_{T\in\mathcal T_1} \prod_{e_i \notin T} x_i \qquad\text{,}\quad \Ff = \sum_{F\in\mathcal T_2} p_F^2 \prod_{e_i\notin F} x_i + \Uu \sum_{i=1}^n x_i m_i^2 \label{eq:SymanzikDef2}
      \end{align}
      where $p_F$ is the sum of momenta flowing from one part of the forest $F$ two the other part. For detailed information we refer to \cite{BognerFeynmanGraphPolynomials2010} and \cite{NakanishiGraphTheoryFeynman1971}. As obviously by (\ref{eq:SymanzikDef2}), Symanzik polynomials are homogeneous polynomials of degree $L$ and $L+1$, respectively. In addition to it the first Symanzik polynomial and the massless part of the second Symanzik polynomial are linear in each Schwinger parameter. Alternatively, we can consider the second Symanzik polynomial as the discriminant of $\Lambda(k,p,x) \Uu(x)$ with respect to $k$, i.e. we obtain the second Symanzik polynomial by eliminating $k$ in $\Lambda(k,p,x)\Uu(x)$ by means of the equation $\frac{\partial \Lambda}{\partial k_j} = 0$.    
      
      In parametric space, there are several representations of the Feynman integral known. The so-called Feynman representation express Feynman integral as
      \begin{align}
        I_\Gamma (d,\nu,p,m) = \frac{\Gamma(\omega)}{\Gamma(\nu)}\int_{\mathbb R^n_+} dx x^{\nu-1} \delta(1-H(x)) \frac{\Uu^{\omega-\frac{d}{2}}}{\Ff^\omega} \label{eq:FeynmanUF}
      \end{align}    
      where $\omega=\sum_{i=1}^n \nu_i - L \frac{d}{2}$ is the superficial degree of divergence and $H(x):= \sum_{i=1}^n h_i x_i$ is an arbitrary hyperplane in $\mathbb R^n$ with $h_i\geq 0$ not all zero. The freedom of the choice of this hyperplane is sometimes referred as Cheng-Wu theorem and expresses the projective nature of the integral (\ref{eq:FeynmanUF}) due to the homogenity of the Symanzik polynomials. For a proof of the equivalence of equation (\ref{eq:FeynmanMomSp}) and (\ref{eq:FeynmanUF}) we refer to \cite{ItzyksonQuantumFieldTheory1980} and \cite{PanzerFeynmanIntegralsHyperlogarithms2015}. In order to simplify the notation, we will use a multiindex notation, which implies in particular $\Gamma(\nu):=\prod_i \Gamma(\nu_i)$, $dx:= \prod_i dx_i$, $x^{\nu-1}:= \prod_i x_i^{\nu_i-1}$.
      
      Another parametric integral representation is due to Lee and Pomeransky \cite{LeeCriticalPointsNumber2013}
      \begin{align}
        I_\Gamma (d,\nu,p,m) = \frac{\Gamma\left(\frac{d}{2}\right)}{\Gamma(\nu)\Gamma\left(\frac{d}{2}-\omega\right)}\int_{\mathbb R^n_+} dx x^{\nu-1} \Gg^{-\frac{d}{2}} \label{eq:FeynmanG}
      \end{align}
      where $\Gg:=\Uu+\Ff$ is the sum of the first and the second Symanzik polynomial. We refer to \cite{BitounFeynmanIntegralRelations2019} for a proof. \\
      
      Instead of the external momenta $p$ and the masses $m$, the parametric representations indicate another choice of what we want to use as the variables of the Feynman integral. Thus, we want to use the coefficients of the Symanzik polynomials as the variables of Feynman integrals. Hence, we can write
      \begin{align}
        \Gg = \sum_{a\in A} z_j x^a \label{eq:Gsupport}
      \end{align}
      where $A\subset\mathbb Z^n$ is the set of exponents and $z\in\mathbb C^N$ are the new variables of the Feynman integral. To avoid redundancy we will always assume that $z_j\not\equiv 0$ and that all elements of $A$ are pairwise disjoint.
      
      In equation (\ref{eq:Gsupport}) we have introduced in fact a generalization of Feynman integrals by way of the back door. Equation (\ref{eq:Gsupport}) gives also coefficients to the first Symanzik polynomial and it is also implicitly assumed that the coefficients in the second Symanzik polynomial are independent of each other. We will call such a generalization of Feynman integral a \textit{generic Feynman integral}. It is obvious that we can consider the physical relevant case by specifying the variables $z$ in the generic Feynman integral. However, as we discuss in the following, such a limit from generic Feynman integrals to physical Feynman integrals will not always be unproblematic.\\
      
      The parametric representations (\ref{eq:FeynmanUF}) and (\ref{eq:FeynmanG}) belong both to the class of Euler-Mellin integrals \cite{BerkeschEulerMellinIntegrals2013}, which are defined as Mellin transforms of a product of polynomials up to some powers. As every Euler-Mellin integral, also Feynman integrals are $\Aa$-hypergeometric functions. For Feynman integrals this fact was first noted as a sidemark by Gelfand, Kapranov, Zelevinsky themselves \cite{GelfandHypergeometricFunctionsToral1989} and later implicitly used in \cite{NasrollahpoursamamiPeriodsFeynmanDiagrams2016}. Explicitly, it was shown independently in \cite{DeLaCruzFeynmanIntegralsAhypergeometric2019} and \cite{KlausenHypergeometricSeriesRepresentations2019}.
      
      \begin{theorem}[$A$-hypergeometric Feynman integrals \cite{KlausenHypergeometricSeriesRepresentations2019}]
        Let $A\subset\mathbb Z^n$ be the support of the sum of the Symanzik polynomials $\Gg=\Uu+\Ff$ according to (\ref{eq:Gsupport}). By $\Aa$ we denote the homogenization of $A$, i.e. we interpret $A$ as a set of column vectors building an $n\times N$ integer matrix and adding the row $(1,\ldots,1)$. Further, let $\nuu =(\nu_0,\nu)\in\mathbb C^{n+1}$ be the combination of the space-time dimension $\nu_0:=\frac{d}{2}$ and the indices $\nu$. Then every scalar Feynman integral $I_\Gamma(\nuu,z)$ without tadpoles\footnote{A ``tadpole'' in a Feynman graph means a loop consisting in only one edge. Note that the common nomenclature in graph theory differs slightly from the nomenclature of Feynman graphs.} satisfies all differential equations of the $\Aa$-hypergeometric system $H_\Aa(\nuu)$
        \begin{align}
          & E_i(\nuu) \bullet I_\Gamma(\nuu,z) = 0 \quad\text{for}\quad i=0,\ldots,n \\
          & \square_l \bullet I_\Gamma(\nuu,z) = 0 \quad\text{for all}\quad l\in\mathbb L = \ker_{\mathbb Z}(\Aa)\point
        \end{align}
      \end{theorem}
      
      Also for the representation in (\ref{eq:FeynmanUF}) we can establish an $\Aa$-hypergeometric system in the following way. W.l.o.g. we will set $x_n=1$ by evaluating the delta distribution in (\ref{eq:FeynmanUF}). Denote by $A_\Uu$ and $A_\Ff$ the support of the first and the second Symanzik polynomial after setting $x_n=1$. In doing so, we can construct the following matrix      
      \begin{align}
        \Aa^\prime = \begin{pmatrix}
          1 & \cdots & 1 & 0 & \cdots & 0\\
          0 & \cdots & 0 & 1 & \cdots & 1\\
          \\
            & A_\Uu      &   &   &  A_\Ff &\\
            \\            
        \end{pmatrix}
      \end{align}
      which defines together with $\beta = \left(\frac{d}{2}-\omega,\omega,\nu_1,\ldots,\nu_{n-1}\right)$ the $\Aa$-hypergeometric system $H_{\Aa^\prime}(\beta)$ of (\ref{eq:FeynmanUF}). As expected the $\Aa$-hypergeometric systems for (\ref{eq:FeynmanUF}) and (\ref{eq:FeynmanG}) are equivalent, which can be verified by the unimodular matrix
      \begin{align}
        T = \left(\begin{array}{cccc}
                L+1 & -1 & \cdots & -1 \\
                -L  & 1  & \cdots & 1  \\
                0   &    &        & 0 \\  
                \vdots &   \multicolumn{2}{c}{\smash{\scalebox{1.5}{$\mathds{1}$}}}       & \vdots \\
                0      & &        & 0
            \end{array}\right)\quad\text{,}\qquad 
        T^{-1} =  \left(\begin{array}{ccccc}
                1 & 1 & 0 & \cdots & 0 \\
                0 & 0 & & & \\
                \vdots & \vdots &   \multicolumn{3}{c}{\smash{\scalebox{1.5}{$\mathds{1}$}}} \\
                0 & 0 & & \\
                L & L+1 & -1 & \cdots & -1
            \end{array}\right)
      \end{align}
      which transforms $\Aa^\prime = T \Aa$ and $\beta = T \nuu$. In the following we will mostly prefer the Lee-Pomeransky representation (\ref{eq:FeynmanG}) due to its plainer structure.

      \begin{remark}
        The correspondence between the polynomial $\Uu+\Ff$ and the polynomial $(\Uu \cdot \Ff)|_{x_n=1}$ which appears in the Lee-Pomeransky and in the Feynman representation, respectively, will be natural in the light of equivalence classes in Grothendieck rings. Without going into detail, we will refer to \cite[lem. 48 and cor. 49]{BitounFeynmanIntegralRelations2019}. This general correspondence leads also to the property that the $A$-discriminants and principal $A$-determinants of $\Uu+\Ff$ and $(\Uu \cdot \Ff)|_{x_n=1}$ coincide. The latter is also apparent by the fact, that the $A$-discriminants only depend on the affine geometry of the point configuration $A$.
      \end{remark}

    \subsection{Singularities of the Feynman integral}
    
      In the previous section we introduced several integrals, as representations of a more general function depending on parameters $\nuu=\left(\frac{d}{2},\nu\right)$ and variables $z$, which encode the dependence of masses $m$ and external momenta $p$. Thus, one of the first natural questions is whether these integrals will converge. If we exclude tadpole graphs\footnote{In the renormalization procedure, one can add a counter term, which removes precisely all tadpole contributions \cite{StermanIntroductionQuantumField1993}.}, it is well-known, that there always exists a complex domain of parameters $\nuu$ and variables $z$, such that the integral representations converge absolutely. For all $\Re z_i > 0$, which we will call the Euclidean sector, the Feynman integral in representation (\ref{eq:FeynmanG}) converges absolutely if and only if $\nu \in \relint \left(\frac{d}{2} \Newt(\Gg) \right)$ is in the relative interior of the dilated Newton polytope of $\Gg$ \cite{BerkeschEulerMellinIntegrals2013, SchultkaToricGeometryRegularization2018, KlausenHypergeometricSeriesRepresentations2019}. For the other representations of Feynman integrals we can derive similar conditions. 
      
      However, these regions of convergence will only cover a small part of the domain where we can analytically continue the Feynman integral representations, even though the Feynman integral is not an entire function. We can distinguish two different kind of singularities which will appear in the analytic continuation: singularities in the parameters $\nuu$ and singularities in the variables $z$. For the parameters $\nuu$ these singularities are known as UV- and IR-divergences and the Feynman integral has only poles in these parameters $\nuu$. Moreover, one can simply describe the possible poles in $\nuu$ by the facets of the Newton polytope $\Newt(\Uu+\Ff)$ by a clever use of integration by parts.
      
      \begin{theorem}[Meromorphic continuation \cite{KlausenHypergeometricSeriesRepresentations2019,BerkeschEulerMellinIntegrals2013}]
        Describe the Newton polytope $\Newt(\Gg)$ as an intersection of half-spaces according to equation (\ref{eq:HPolytope}). Then every non-tadpole Feynman integral $I_\Gamma(\nuu,z)$ with appropriate values $z$ can be written as
        \begin{align}
          I_\Gamma(\nuu,z) = \Phi_\Gamma(\nuu,z) \frac{\prod_{j=1}^k \Gamma( b_j \Re \nu_0 - m_j^T \cdot \Re \nu)}{\Gamma(\nu_0-\omega)\Gamma(\nu)}
        \end{align}
        where $\Phi_\Gamma(\nuu,z)$ is an entire function with respect to $\nuu\in\mathbb C^{n+1}$. As before we use $\nu_0 = \frac{d}{2}$ and $\nuu = (\nu_0,\nu)$.
      \end{theorem}
      
      Hence, we can continue the Feynman integral meromorphically with respect to its parameters $d,\nu$ and we can easily give a necessary condition for its poles.\\
      
      Considerably more difficult is the situation for the variables $z$ of the Feynman integral. We will find certain combinations of $p$ and $m$, such that the Feynman integral has lacking analyticity or differentiability. Considering the momentum space representation (\ref{eq:FeynmanMomSp}), those singularities may appear if some inverse propagators $D_i$ vanish and additionally the integration contour is trapped in such a way, that we are not able to elude the singularity by deforming the contour in the complex plane. These situations are called pinches and if they appear the equations
      \begin{align}
        x_i D_i = 0 \qquad \text{for all} \quad i=1,\ldots,n \label{eq:MomLandau1}\\
        \frac{\partial}{\partial k_j} \sum_{i=1}^n x_i D_i = 0 \qquad \text{for all} \quad j=1,\ldots,L \label{eq:MomLandau2}
      \end{align}
      have a solution for $x\in\mathbb C^n\setminus\{0\}$ and $k\in\mathbb C^{L\times d}$. We will call all points $z$ admitting such a solution a \textit{Landau critical point}. Landau critical points do not depend on the choice of internal momenta or their orientation. The equations (\ref{eq:MomLandau1}), (\ref{eq:MomLandau2}) are called \textit{Landau equations} and were independently derived in 1959 by Bjorken \cite{BjorkenExperimentalTestsQuantum1959}, Landau \cite{LandauAnalyticPropertiesVertex1959} and Nakanishi \cite{NakanishiOrdinaryAnomalousThresholds1959}. We recommend \cite{MizeraCrossingSymmetryPlanar2021} for a comprehensive summary of the known research results in Landau's analysis of over 60 years. Further we will refer to \cite{EdenAnalyticSmatrix1966} for a classical and \cite{MuhlbauerMomentumSpaceLandau2020, CollinsNewCompleteProof2020} for modern derivations of Landau equations. 
      
      Unfortunately, strictly speaking the Landau equations (\ref{eq:MomLandau1}), (\ref{eq:MomLandau2}) are neither necessary nor sufficient conditions to have a singularity of the analytic continuated Feynman integral. Thus, there are on the one hand singularities which does not correspond to a solution of Landau equations. Those singularities are often called second-type singularities or non-Landau singularities and were found for the first time in \cite{CutkoskySingularitiesDiscontinuitiesFeynman1960}. And on the other hand, not all solutions of Landau equations resulting in a singularity \cite{CollinsNewCompleteProof2020}. However, Landau equations are necessary and sufficient for the appearance of a trapped contour \cite{CollinsNewCompleteProof2020} and can be necessary for certain restrictions on Feynman integrals. We will call the singularities coming from Landau equations anomalous thresholds, except for those singularities corresponding to unitarity cuts, which we will call normal thresholds. Furthermore, singularities with all $x_i\neq 0$ in (\ref{eq:MomLandau1}), (\ref{eq:MomLandau2}) will be called leading singularities and we will usually distinguish between solutions with real positive $x_i\geq 0$ and general complex $x_i\in\mathbb C$. 
      
      From a physical perspective, Landau equations determine when internal (virtual) particles going on-shell. Hence, the Feynman diagram describes then an interaction between real particles with a specific lifetime \cite{ColemanSingularitiesPhysicalRegion1965}.
     
      The extraordinary meaning for Feynman integrals owing the Landau singularities also from various methods, which construct the whole Feynman integral from these singularities. All these methods root more or less in the optical theorem and the corresponding unitarity cuts, introduced by Cutkosky \cite{CutkoskySingularitiesDiscontinuitiesFeynman1960} shortly after Landau's article. However, it should be mentioned that Cutkosky's rules are unproven up today. We refer to \cite{BlochCutkoskyRulesOuter2015} for the recent progress of giving a rigorous proof of Cutkosky's rules. Also for other techniques, as for example in sector decomposition, Landau's analysis plays an important role.
      
      The Landau equations stated in (\ref{eq:MomLandau1}), (\ref{eq:MomLandau2}) involve the integration variables in momentum space as well as the integration variables of parametric space. There are also equivalent equations, which are stated in the parametric variables only. Since the second Symanzik polynomial can be written as a discriminant of $\Lambda \Uu$ with respect to the loop momenta $k$, it is immediately clear, that those equations will be conditions on the second Symanzik polynomial $\Ff$. Instead of eliminating $k$ from (\ref{eq:MomLandau1}), (\ref{eq:MomLandau2}), we can show the Landau equations in parametric space also directly by considering the parametric integral representations \cite{NakanishiOrdinaryAnomalousThresholds1959}.
      
      \begin{theorem}[Parametric space Landau equations e.g. \cite{NakanishiGraphTheoryFeynman1971, EdenAnalyticSmatrix1966}] \label{thm:ParLandau}
        Under the assumption $\Uu\neq 0$, a point $z\in\mathbb C^N$ is a Landau critical point, if and only if the equations
        \begin{enumerate}
          \item $x_i \frac{\partial \Ff}{\partial x_i} = 0$ for $i=1,\ldots,n$
          \item $\Ff=0$
        \end{enumerate}
        have a solution in $x\in\mathbb P^{n-1}_{\mathbb C}$. The case $\Uu=0$ is connected with the second-type singularities, which we will examine later.
      \end{theorem}
      \begin{proof}
        ``$\Rightarrow$'':
        Consider $\Lambda$ from equation (\ref{eq:LandauPsi}). We will find
        \begin{align}
          \frac{\partial \Lambda(k,p,x)}{\partial k} &= 2 Mk + 2 B^T p \point
        \end{align}
        By the assumption $\Uu\neq 0$, $M$ is invertible and thus $\pd{\Lambda}{k}=0$ implies $k=-M^{-1}B^T p$. Inserting this equation for $k$ in  (\ref{eq:LandauPsi}) and comparing with (\ref{eq:SymanzikDef1}) we will find $\Lambda(-M^{-1}B^Tp,p,x) = \Ff/\Uu$. Therefore, $\Lambda=0$ implies $\Ff=0$ and furthermore
        \begin{align}
          x_j \frac{\partial \Ff}{\partial x_j} = x_j \frac{\partial}{\partial x_j} \left( \Uu \Lambda \right) = \Uu x_j D_j = 0 \point
        \end{align}
        ``$\Leftarrow$'': Since the definition (\ref{eq:MomLandau1}),(\ref{eq:MomLandau2}) contains more variables as in parameter space, we can always find a value $k^\prime$, s.t. $\Lambda(k^\prime,p,x) = \Ff/\Uu$, without restricting the possible solutions for $x$. This can also be vindicated by the fact that the Feynman integral (\ref{eq:FeynmanMomSp}) is invariant under linear transformations. We conclude
        \begin{align}
          \frac{\partial\Lambda}{\partial k'} &=  2Mk^\prime + 2B^Tp = 0\\
           x_j D_j &= x_j \frac{\partial \Lambda}{\partial x_j} = x_j \frac{\partial \Uu^{-1}}{\partial x_j} \Ff + x_j \frac{\partial \Ff}{\partial x_j} \Uu = 0 \point
        \end{align}
      \end{proof}
      
      Note, that by Euler's theorem one of the $n+1$ equations in theorem \ref{thm:ParLandau} is redundant, which is the reason why we look for solutions in projective space. \\
      
      According to \cite{BrownPeriodsFeynmanIntegrals2010, PhamSingularitiesIntegrals2011, MuhlbauerMomentumSpaceLandau2020} we will call the variety defined by the closure of all Landau critical points the \textit{Landau variety} $\mathcal L(I_\Gamma)$. Due to Riemann's second removable theorem \cite{KaupHolomorphicFunctionsSeveral1983}, we are especially interested in the codimension one part of $\mathcal L(I_\Gamma)$, which we will denote by $\mathcal L_1(I_\Gamma)$. Based on the Landau equations in parameter space, we can directly read off the following theorem from the definition of the principal $A$-determinant.
    
      \begin{theorem}[Landau variety] \label{thm:LandauVar}
        Let $\Ff\in\mathbb C[x_1,\ldots,x_n]$ be the second Symanzik polynomial of a Feynman graph $\Gamma$ and let $A_\Ff\subset \mathbb Z^n$ be the support of $\Ff$. The Landau variety is given by the (simple) principal $A$-determinant of $\Ff$
        \begin{align}
          \mathcal L_1(I_\Gamma) = \mathbf V( E_{A_\Ff}(\Ff)) =  \mathbf V( \widehat E_{A_\Ff}(\Ff)) \point
        \end{align}
      \end{theorem}
      
      However, we have to mention that in the definition of the Landau variety as well as in theorem \ref{thm:LandauVar} we will consider generic values $z\in\mathbb C^N$ of the coefficients in the second Symanzik polynomial. In the physical relevant case, there are relations among these coefficients and they are not necessarily generic enough. This is not only an issue in parametric representation, which involves the Symanzik polynomials. It also appears in momentum space, where the external momenta are treated as vectors in $d$-dimensional Minkowskian space. Thus, there can not be more than $d$ independent external momenta and additionally we suppose an overall momentum conservation. If the variables are not generic enough it can occur that the defining equation of the Landau variety is identical to zero. Thus, the Landau variety would be equal to the whole space $\mathbb C^N$. On the other hand, we know that there can not be a singularity with unbounded functional value for all points $z\in\mathbb C^N$ due to the convergence considerations from the previous section. Thus, in the limit to the physical relevant case, we want to exclude such that ``overall singularities'' in order to make the other singularities apparent. Therefore, we want to define
      \begin{align}
        \widehat E_{A_\Ff}^{ph}(\Ff) = \prod_{\substack{\tau\subseteq\Newt(\Ff) \\ \left.\Delta_{A_\Ff\cap\tau}(\Ff_\tau)\right|_{z\rightarrow z^{(ph)}} \neq 0}} \Delta_{A_\Ff\cap\tau}(\Ff_\tau)
      \end{align}
      a simple principal $A$-determinant which contains only the physical relevant parts of the principal $A$-determinant, i.e. we omit the parts, which vanish after inserting the physical restrictions $z^{(ph)}$ on the variables. Equivalently, we define $\mathcal L_1^{ph}(I_\Gamma) := \mathbf V( \widehat E_{A_\Ff}^{ph}(\Ff))$ as the physical relevant Landau variety. Note that the Landau variety $\mathcal L_1(I_\Gamma)$ is independent from the parameters $\nuu$, whereas the physical Landau variety $\mathcal L_1^{ph}(I_\Gamma)$ can depend on the parameters $\nuu$. For example specific choices of $d$ can change the relations between the external momenta.\\
      
      Usually, one splits the calculation of the Landau singularities in a leading singularity with all $x_i\neq 0$ and singularities with $x_i=0$ for $i\in I$, where $\emptyset\neq I\subsetneq \{1,\ldots,n\}$. The latters can be considered as leading singularities of subgraphs, since setting a Schwinger parameter $x_i=0$ in the second Symanzik polynomial to zero is the same as considering a subgraph, where the corresponding edge $e_i$ is shrinked. However, the principal $A$-determinant obeys a different splitting according to theorem \ref{thm:pAdet-factorization}. These decompositions coincide if the second Symanzik polynomial consists in all monomials of a given degree.
      
      \begin{lemma}
        For an index set $I\subseteq\{1,\ldots,n\}$ we call $\Ff_I(x):=\Ff(x)|_{\{x_i=0\}_{i\in I}}$ the subgraph polynomial associated to $I$. Every subgraph polynomial is also a truncated polynomial $\Ff_\tau$ with a face $\tau\subseteq\Newt(\Ff)$. The converse is true if $\Ff$ consists in all monomials of degree $L+1$. However, the converse is not true in general.
      \end{lemma}
      \begin{proof}
        Choose $\phi(v) = - \sum_{i=1}^n b_i v_i$ with $b_i=1$ for all $i\in I$ and $b_i=0$ otherwise. This linear function takes its maximal value  $\max \phi(v) =0$ for precisely those values $v\in\mathbb R_+^n$ with $v_i=0$ for $i\in I$. Since all points of $\Newt(\Ff)$ are contained in the positive orthant $\mathbb R_+^n$, such a linear map $\phi$ defines the corresponding face $\tau$ according to (\ref{eq:facedef}).
        
        In case, where $\Ff$ consists in all possible monomials of degree $L+1$ the Newton polytope is an $n$-simplex.
      \end{proof}
      
      Thus, beyond $1$-loop graphs and banana graphs, which contain all monomials of a given degree in the second Symanzik polynomial, one may gets additional singularities from the truncated polynomials, which will be missed with the approach of the subgraphs. Although, considering the subgraphs seems evident at first sight, the structure of Landau varieties is somehow more subtle. These subtleties emerge mainly from the Zariski-closure in the definition of Landau varieties. \\
      
      Comparing theorem \ref{thm:LandauVar} with the results of section \ref{sec:AHyp} about $A$-hypergeometric functions, we would rather expect $\mathbf V(E_{A_\Gg}(\Uu+\Ff))$ instead of $\mathcal L_1(I_\Gamma)$ to be the singular locus of $I_\Gamma$. Directly from the factorization of the principle $A$-determinant we can see the relation of these two varieties.
      \begin{lemma}
        The Landau variety is contained in the singular locus of the $A$-hypergeometric function
        \begin{align}
          \mathbf V\big(E_{A_\Ff}(\Ff)\big) \subseteq \mathbf V \big(E_{A_\Gg}(\Uu+\Ff)\big)\point
        \end{align}
      \end{lemma}
      \begin{proof}
        $\Uu$ and $\Ff$ are homogeneous polynomials of different degrees. Therefore, $\Newt(\Uu+\Ff)$ has points on two different, parallel hyperplanes and thus $\Newt(\Uu)$ and $\Newt(\Ff)$ are two facets of $\Newt(\Uu+\Ff)$. By theorem \ref{thm:pAdet-factorization} we see that $\mathbf V (E_{A_\Gg}(\Uu+\Ff)) = \mathbf V (E_{A_\Ff}(\Ff)) \cup \mathbf V (E_{A_\Uu}(\Uu)) \cup \Delta_\Gg(\Uu+\Ff) \cup \mathbf V(R)$, where the remaining polynomial $R$, correspond to all discriminants coming from proper, mixed faces, i.e. faces $\tau\subsetneq \Newt(\Gg)$ having points of $\Uu$ and $\Ff$.
      \end{proof}
    
      Thus, we obtain what we already expected: The Landau variety covers not all kinematic singularities of the Feynman integrals and in general $\mathbf V(E_{A_\Ff}(\Ff))$ will be a proper subvariety of $\mathbf V (E_{A_\Gg}(\Gg))$. Based on the prime factorization of the principal $A$-determinant, we will divide the singular locus of the Feynman integral $\mathbf V (E_{A_\Gg}(\Gg))$ into four parts
      \begin{align}
         \mathbf V \big(E_{A_\Gg}(\Gg)\big) = \mathbf V \big(E_{A_\Ff}(\Ff)\big) \cup \mathbf V \big(E_{A_\Uu}(\Uu)\big) \cup \mathbf V \big(\Delta_{A_\Gg}(\Gg)\big) \cup \mathbf V (R) \point
      \end{align}
      Namely, we decompose $E_{A_\Gg}(\Gg)$ in a polynomial $E_{A_\Ff}(\Ff)$ generating the classical Landau variety according to theorem \ref{thm:LandauVar}, a polynomial $E_{A_\Uu}(\Uu)$ which is constant in the physical relevant case and a polynomial $\Delta_{A_\Gg}(\Gg)$, which we will associate to the second-type singularities. The remaining polynomial
      \begin{align}
         R := \prod_{\substack{\tau\subsetneq \Newt(\Uu+\Ff)\\ \tau \nsubseteq \Newt (\Uu), \tau \nsubseteq \Newt (\Ff)}} \Delta_{A\cap\tau} (\Gg_\tau) 
      \end{align}
      will correspond to second-type singularities of subgraphs and we will call the roots of $R$ the \textit{mixed type singularities of proper faces}.
      
      In the following section, we will analyze step by step these further contributions to the singular locus.
    
    \subsection{Second-type singularities}
   
      As aforementioned the defining polynomial of the singular locus $\Sing(\mathcal M_\Aa(\nuu))$ splits into several discriminants.  With the $A$-discriminant $\Delta_{A_\Gg}(\Uu+\Ff)$ we will associate the so-called second-type singularities. We have to remark, that the notion of second-type singularities differs slightly in various literature. Moreover, there is very little known about second-type singularities. Usually, there is been made a distinction between pure second-type singularities and mixed second-type singularities \cite{EdenAnalyticSmatrix1966, NakanishiGraphTheoryFeynman1971}. The pure second-type singularities do not depend on masses and can be expressed by Gram determinants, whereas the latters appear in higher loops. Second-type singularities are better understood in momentum space, whereas they are endpoint singularities at infinity \cite{MuhlbauerMomentumSpaceLandau2020}. In parametric space, second-type singularities are connected to the case where $\Uu=0$.\\
      
      By introducing a new variable $x_0$, we can change to the homogeneous setting $\Delta_{\tilde A_\Gg}(x_0 \Uu + \Ff)$ which has the same discriminant, since there is an appropriate injective, affine map connecting $A_\Gg$ with $\tilde A_\Gg$ according to section \ref{sec:Adisc}. Writing the corresponding polynomial equation system explicitly, the $A$-discriminant $\Delta_{A_\Gg}(\Gg)$ is the defining polynomial of the closure of coefficients $z\in\mathbb C^N$, such that the equations
      \begin{align}
        \Uu = 0, \quad \Ff_0 = 0, \quad \pd{\Ff_0}{x_i} + \pd{\Uu}{x_i} \left( x_0 + \sum_{j=1}^n x_j m_j^2 \right) = 0 \quad\text{for}\quad i=1,\ldots,n \label{eq:secondtype} \point
      \end{align}
      have a solution for $(x_0,x)\in(\mathbb C^*)^{n+1}$.  By $\Ff_0 = \sum_{F\in\mathcal T_2} p_F^2 \prod_{e_i\notin F} x_i$ we denote the massless part of the second Symanzik polynomial. Not all conditions herein (\ref{eq:secondtype}) are independent, since the polynomial $x_0\Uu+\Ff$ is homogeneous again. Thus, we can drop an equation from (\ref{eq:secondtype}). These equations for second-type singularities (\ref{eq:secondtype}) agree with the result in \cite{NakanishiGraphTheoryFeynman1971}.
      
      \begin{example}[\nth{2} type singularities of all banana graphs]
        Consider the family of massive $L$-loop $2$-point functions, which are also called banana graphs. These graphs having $n=L+1$ edges and the Symanzik polynomials
        \begin{align}
          \Uu &= x_1 \cdots x_n \left(\frac{1}{x_1} + \ldots + \frac{1}{x_n} \right) \\
          \Ff_0 &= p^2 x_1\cdots x_n \point
        \end{align}
        Applying the conditions of (\ref{eq:secondtype}) we will find the second-type singularity for all banana graphs to be
        \begin{align}
          p^2 = 0 \point
        \end{align}
      \end{example}
      \begin{example}[\nth{2} type singularities of all $1$-loop graphs] \label{ex:1loopSecondtype}
        A massive $1$-loop graph with $n$ edges has Symanzik polynomials
        \begin{align}
          \Uu &= x_1 + \ldots + x_n \\
          \Ff_0 &= \sum_{1\leq i < j \leq n} s_{ij} x_i x_j 
        \end{align}
        where $s_{ij} := \left(\sum_{k=i}^{j-1} p_k\right)^2$ for $i<j$ defines the dependence on external momenta. We will set $s_{ij}=s_{ji}$ for $i>j$ and $s_{ii}=0$. We obtain $\pd{\Uu}{x_j} = 1$ and $\pd{\Ff_0}{x_j} = \sum_{i\neq j} s_{ij} x_i$ for the derivatives. Since we can drop one equation from (\ref{eq:secondtype}) due to the homogenity of $x_0 \Uu +\Ff$, we will obtain a linear equation system in the $1$-loop case. Eliminating $x_0$ by subtracting equations, we can combine these conditions to a determinant
        \begin{align}
          \begin{vmatrix}
            \hspace{.3cm} 1 & 1 & \cdots & 1 & \\
            \multicolumn{4}{c}{\multirow{2}{*}{$\bigl(s_{ij} - s_{jn}\bigr)_{\scalebox{0.7}{$\substack{1\leq i \leq n-1 \\ 1\leq j \leq n}$}}$}} \\
            &
          \end{vmatrix} = 0 \label{eq:1loopSecondtype}
        \end{align}
        as the (not necessarily irreducible) defining polynomial of the second-type singularity. By the same argument as used in \cite{NakanishiGraphTheoryFeynman1971}, the condition (\ref{eq:1loopSecondtype}) is equivalent to the vanishing of the Gram determinant, which is usually set for the pure second-type singularity \cite{EdenAnalyticSmatrix1966}. Note that in the $1$-loop case, the second-type singularity does not depend on masses. Furthermore, for higher $n$, the external momenta satisfies certain relations, since there can not be more than $d$ linear independent vectors in $d$-dimensional Lorentz space. In addition the external momenta ensures a conservation law. Thus, for $n\geq d + 2$ the condition (\ref{eq:1loopSecondtype}) is satisfied for all physical external momenta. Hence, we will remove these contribution to the singular locus, when we restrict us to the physical relevant case.
      \end{example} 
      
      As the next part of the singular locus $\Sing(\mathcal M_\Aa(\nuu))$ we will consider the principal $A$-determinant of the first Symanzik polynomial $E_{A_\Uu}(\Uu)$. Since the coefficients of the first Symanzik polynomial are all equal to one, the principal $A$-determinant $E_{A_\Uu}(\Uu)$ can be either zero or one. In the first case, we will obtain similar to the example \ref{ex:1loopSecondtype} an always present contribution to the singular locus. That this constant contribution can not correspond to an actual singularity can be seen analogously, by the existence of a non-vanishing convergence domain of the Feynman integral. Therefore, we will exclude this contribution to the singular locus if we consider the physical relevant case. We want to emphasize that these removing of ``the unwanted zeros'' shows not only up for our special approach. Such a phenomenon does also appear in the ``classical way'' of treating Landau singularities, as for example in the previous discussed second-type singularities of $1$-loop graphs when the number of legs exceeds the spacetime dimension $d+1$. 

      \begin{lemma}
        For all $1$-loop graphs and all $L$-loop banana graphs we obtain
        \begin{align}
          E_{A_\Uu}(\Uu) = 1 \point
        \end{align}
      \end{lemma}
      \begin{proof}
        Since $\Uu$ has no free coefficients, $E_{A_\Uu}(\Uu)$ is either $0$ or $1$. In both cases $\Newt(\Uu)$ describes a simplex. Therefore, in order to proof $E_{A_\Uu}(\Uu)=1$ it is sufficient to show, that $\Uu=\pd{\Uu}{x_1} = \ldots=\pd{\Uu}{x_n}=0$ contains a contradiction. For $1$-loop graphs this contradiction is obvious since $\pd{\Uu}{x_i}=1$.
        
        For $L$-loop banana graphs consider
        \begin{align}
          \Uu - x_i \pd{\Uu}{x_i} = - \frac{x_1 \cdots x_n}{x_i} = 0
        \end{align}
        which has no solution for $x\in(\mathbb C^*)^n$.
      \end{proof}
      
      We want to emphasize, that this lemma is not true for general Feynman graphs. The dunce's cap graph is the simplest example, which will have a vanishing principal $A$-determinant of the first Symanzik polynomial. Thus, in those cases the physical Feynman integral is a (potential) branch point of a generic Feynman integral. Simultaneously, the dimension of solution space of a physical Feynman integrals is smaller than the dimension of the solution space of generic Feynman integrals. However, we can always consider in principle the limit from generic to physical Feynman integrals. In \cite{KlausenHypergeometricSeriesRepresentations2019} it was indicated how such a limit can be computed in general. In our approach we want always assume the generic Feynman integral first and the limit to the physical case afterwards in order to be able to use theorem \ref{thm:CKKT}. Vice versa, we could also specifying the $\mathcal D$-module to the physical relevant case first and considering the characteristic variety second. However, it is an algorithmically challenging task to restrict $\mathcal D$-ideals to specific values \cite{SaitoGrobnerDeformationsHypergeometric2000}. \\
      
      The discriminants of proper mixed faces $R$ can be associated with second-type singularities of subgraphs as it can be observed in the following example.
      
      \begin{example}[$R$ for the triangle graph]
        We will determine the discriminants of proper mixed faces of the triangle graph. The Symanzik polynomials for the triangle graph are
        \begin{align}
          \Uu &= x_1 + x_2 + x_3\\
          \Ff &= p_1^2 x_2 x_3 + p_2^2 x_1 x_3 + p_3^2 x_1 x_2 + \Uu (x_1 m_1^2 + x_2 m_2^2 + x_3 m_3^2) \point
        \end{align}
        The corresponding Newton polytope $\Newt(\Uu+\Ff)$ has in total $21$ faces, where one face is the Newton polytope $\Newt(\Uu+\Ff)$ itself, $7$ faces corresponding to faces of $\Newt(\Uu)$, another $7$ corresponding to faces of $\Newt(\Ff)$ and the remaining $6$ faces are proper mixed faces. The truncated polynomials corresponding to these $6$ faces are up to permutations  $1\leftrightarrow 2 \leftrightarrow 3$ given by
        \begin{align}
          r_1 &= m_2^2 x_2^2 + m_3^2 x_3^2 + x_2 + x_3 + (p_1^2+m_2^2+m_3^2) x_2 x_3 \\
          h_1 &= m_1^2 x_1^2 + x_1 \point
        \end{align}
        As $h_1 = \pd{h_1}{x_1} = 0$ has no common solution we have $\Delta(h_1) = 1$ for the $A$-discriminant of $h_1$ as well as for the symmetric permutations. Considering $r_1=\pd{r_1}{x_2}=\pd{r_1}{x_3} = 0$, this is nothing else than the sum of Symanzik polynomials of the $1$-loop bubble graph and we obtain $p_1^2=0$ and similarly for all the permutations. Thus, we get
        \begin{align}
          R = p_1^2 p_2^2 p_3^2
        \end{align}
        as the contribution to the singular locus of the triangle graph from proper mixed faces. 
      \end{example}
      
      We want to mention, that also for the discriminants of proper mixed faces $R$ there can be contributions which are identically zero in the physical relevant case. The simplest example where such a behaviour appears is the $2$-loop $2$-point function also known as the sunset graph. Again, we will define a physical relevant singular locus, where we have to remove these overall contributions.
      
      Although, one does not usually consider the contribution of $R$ to the singular locus of Feynman integrals in literature, we want to remark, that they can give a non-trivial contribution to the singular locus.

  \section{Coamoebas} \label{sec:coamoebas}
  
    In the previous section we only were asking if there are values of $p$ and $m$ such that there exists a multiple root of the Symanzik polynomials anywhere in the complex plane. However, when applied to Feynman integral representations, e.g. the Lee-Pomeransky representation (\ref{eq:FeynmanG}) (or more generally to Euler-Mellin integrals) we only should worry if these multiple roots lying inside the integration region $\mathbb R^n_+$. Thus, in this section we want to study the position of singularities in the space of Schwinger parameters $x$. Hence, we will also consider the multivalued structure of Feynman integrals. In order to better investigate the nature of the multiple roots, we will introduce the concept of coamoebas. Coamoebas were introduced by Mikael Passare as related objects to amoebas, which go back to Gelfand, Kapranov and Zelevinsky \cite{GelfandDiscriminantsResultantsMultidimensional1994}. Amoebas as well as coamoebas provide a link between algebraic geometry and tropical geometry \cite{NisseAmoebasCoamoebasLinear2016}. The relations between coamoebas and Euler-Mellin integrals were investigated in \cite{NilssonMellinTransformsMultivariate2010, BerkeschEulerMellinIntegrals2013}, which the following section is mostly based on. For further reading about coamoebas we refer to \cite{NisseGeometricCombinatorialStructure2009, JohanssonCoamoebas2010, ForsgardHypersurfaceCoamoebasIntegral2012, JohanssonArgumentCycleCoamoeba2013, PassareDiscriminantCoamoebasHomology2012, NisseHigherConvexityCoamoeba2015, ForsgardOrderMapHypersurface2015, ForsgardTropicalAspectsReal2015}. \\
     
    In the following we will consider the more general case of Euler-Mellin integrals, which we will define according to \cite{BerkeschEulerMellinIntegrals2013} as
    \begin{align}
      \mathscr M_f(s,t) := \Gamma(t) \int_{\mathbb R^n_+} dx\, x^{s-1} f^{-t} \label{eq:EulerMellin}
    \end{align}
    where $s\in \mathbb C^n$ and $t\in \mathbb C^k$ are complex parameters. Further, we denote $f^{-t}:= f_1^{-t_1}\ldots f_k^{-t_k}$ for powers of polynomials $f_i\in\mathbb C[x_1,\ldots,x_n]$ and we write $f:=f_1 \cdots f_k$. Parametric Feynman integrals can be treated as a special case of Euler-Mellin integrals with $k=1$ in the Lee-Pomeransky representation (\ref{eq:FeynmanG}) or with $k=2$ in the Feynman representation (\ref{eq:FeynmanUF}) for a convenient hyperplane $H(x)$.
    
    A polynomial $f\in\mathbb C[x_1,\ldots,x_n]$ with fixed coefficients is called \textit{\CNV on} $X\subseteq \mathbb C^n$ if for all faces $\tau\subseteq\operatorname{Newt}(f)$ of the Newton polytope, the truncated polynomials $f_\tau$ do not vanish on $X$. We can consider the vanishing of some truncated polynomial as a necessary condition in the approach of the principal $A$-determinant. Thus, if a polynomial $f$ is \CNV on a set $X$, roots of the principal $A$-determinant $E_A(f)$ will correspond to a solution $x\notin X$ outside of this set $X$. 
    
    Note that the integral (\ref{eq:EulerMellin}) gets ill-defined, when $f$ is not \CNV on the positive orthant $(0,\infty)^n$. Hence, in \cite{BerkeschEulerMellinIntegrals2013, NilssonMellinTransformsMultivariate2010} a slightly more general version of Euler-Mellin integrals was introduced, where we rotate the original integration contour in the complex plane
    \begin{align}
      \mathscr M_f^\theta(s,t) := \Gamma(t) \int_{\Arg^{-1} \theta} dx\, x^{s-1} f(x)^{-t} = e^{is\theta}\, \Gamma(t)  \int_{\mathbb R_+^n} dx\, x^{s-1} f(x\cdot e^{i\theta})^{-t} \label{eq:EulerMellinTheta}
    \end{align}
    with the component-wise argument map $\Arg (x) := (\arg x_1,\ldots,\arg x_n)$. For short we write $f(x\cdot e^{i\theta}) := f(x_1 e^{i\theta_1},\ldots,x_n e^{i\theta_n})$ and we will call (\ref{eq:EulerMellin}) the \textit{$\theta$-analogue Euler-Mellin integral}. Deforming the integration contour slightly in cases where poles of the integrand hit the integration contour, is the same procedure as ``Feynman's $i\epsilon$ prescription''. As aforementioned in section \ref{sec:FeynmanIntegrals} the $i\epsilon$ prescription is equivalent to a redefinition of masses $m_i^2 \mapsto m_i^2 - i\epsilon$, where the sign is motivated by causality \cite{WeinzierlArtComputingLoop2006}. Thus, we obtain 
    \begin{align}
      I_\Gamma^\epsilon (\nuu,z) &= \frac{\Gamma(\nu_0)}{\Gamma(\nu)\Gamma\left(\nu_0-\omega\right)}\int_{\mathbb R^n_+} dx x^{\nu-1} (\Uu+\Ff-i\epsilon \Uu \sum_i x_i)^{-\nu_0} \nonumber \\
       &\stackrel{|\epsilon|\ll 1}{=} \frac{\Gamma(\nu_0)}{\Gamma(\nu)\Gamma\left(\nu_0-\omega\right)}\int_{\mathbb R^n_+} dx x^{\nu-1} (\Uu+\Ff-i\epsilon)^{-\nu_0}
    \end{align}
    for the $i\epsilon$ prescripted Feynman integral, because $\Uu\sum_ix_i >0$ is positive in the orthant $(0,\infty)^n$. Equality of these two integrals is given in the limit $\epsilon\rightarrow 0$. 
    
    By choosing $\theta=(\epsilon,\ldots,\epsilon)$ we get the same result for the $\theta$-analogue Euler-Mellin integral. By the homogenity of the Symanzik polynomials we obtain
    \begin{align}
      I_\Gamma^\epsilon (\nuu,z) &= \frac{\Gamma(\nu_0)}{\Gamma(\nu)\Gamma\left(\nu_0-\omega\right)}\int_{\Arg^{-1}(\epsilon,\ldots,\epsilon)} dx\, x^{\nu-1} (\Uu+\Ff)^{-\nu_0} \nonumber\\
      &=\frac{\Gamma(\nu_0)}{\Gamma(\nu)\Gamma\left(\nu_0-\omega\right)}\int_{\mathbb R^n_+} dx\, x^{\nu-1} e^{i\epsilon(\omega-\nu_0)}(\Uu e^{-i\epsilon} +\Ff)^{-\nu_0} \nonumber\\
       &\stackrel{|\epsilon|\ll 1}{=}  \frac{\Gamma(\nu_0)}{\Gamma(\nu)\Gamma\left(\nu_0-\omega\right)}\int_{\mathbb R^n_+} dx\, x^{\nu-1} (\Uu + \Ff - i\epsilon)^{-\nu_0} \point
    \end{align}
    Thus, in the limit $\epsilon\rightarrow 0$ the $\theta$-analogue Euler-Mellin integral and the $i\epsilon$ prescripted Feynman integral coincide.\\
    
    In order to get a well-defined integral in (\ref{eq:EulerMellinTheta}) we have to track the poles of the integrand. Denote by $\mathcal Z_f := \{ x\in(\mathbb C^*)^n | f(x) = 0\}$ the set of (non-zero) roots of a polynomial $f$. To analyze when these poles meet the integration contour, let
    \begin{align}
      \mathcal C_f := \operatorname{Arg}(\mathcal Z_f) \subseteq \mathbb T^n \label{eq:defCoamoeba}
    \end{align}
    be the argument of the zero locus. We call $\mathcal C_f$ the \textit{coamoeba} of $f$. Since the argument map is a periodic function we restrict the discussion to the $n$-dimensional real torus $\mathbb T^n := (\mathbb R/2\pi\mathbb Z)^n$. The coamoeba is closely related to the \CNV of $f$.
    
    \begin{lemma}[\cite{BerkeschEulerMellinIntegrals2013}]
      For $\theta\in\mathbb T^n$, the polynomial $f(x)$ is \CNV on $\Arg^{-1} (\theta)$ if and only if $\theta \notin \overline{\mathcal C_f}$.
    \end{lemma}
    
    Applied to Feynman integrals, this lemma gives a criterion, whether the $i\epsilon$ prescripted Feynman integral representation (\ref{eq:FeynmanG}) $\lim_{\epsilon\rightarrow 0^+} I^\epsilon_\Gamma(\nuu,z)$ differs from the original Feynman integral $I_\Gamma(\nuu,z)$.
    
    \begin{lemma} \label{lem:nonphysical}
      Assume that $z\in\Sing(\mathcal M_\Aa(\nuu))=\mathbf V(E_A(\Gg))$ is a singular point. If $0\notin \overline{\mathcal C_\Gg}$, the common solution of $G=x_1 \pd{G}{x_1} = \ldots = x_n \pd{G}{x_n} = 0$ generating the singular point according to section \ref{sec:AResultantsDiscriminants} does not lie on the integration contour $x \notin\mathbb R^n_+$.
    \end{lemma}
    
    Thus, we want to distinguish between singular points $z$ corresponding to solutions $x\in\mathbb R_+^n$ and solutions $x\notin\mathbb R^n_+$. The latters are often called unphysical singularities or pseudo thresholds. Those singular points will not lie on the principal sheet (or ``physical sheet'') defined by the integral representation (\ref{eq:FeynmanG}).
    
    Hence, we can use lemma \ref{lem:nonphysical} as a criterion of pseudo thresholds. But we have to remark, that the coamoeba is not an easily accessible object and it will be hard in general to apply lemma \ref{lem:nonphysical} to practical problems. However, coamoebas have many interesting properties, which can be used to give simpler criteria. Without going into details, we will mention a few of the properties. Exemplarily, we have drawn the coamoeba for the massive $1$-loop selfenergy graph for different kinematic regions in figure \ref{fig:coamoebaBubble}.
    
    \begin{figure}[h]
      \begin{center}
      \begin{tabular}{cc}
        \includegraphics[width=5cm]{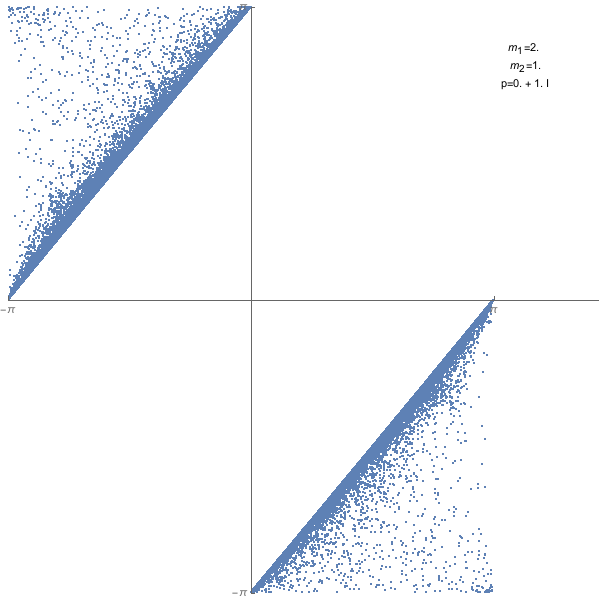} & \includegraphics[width=5cm]{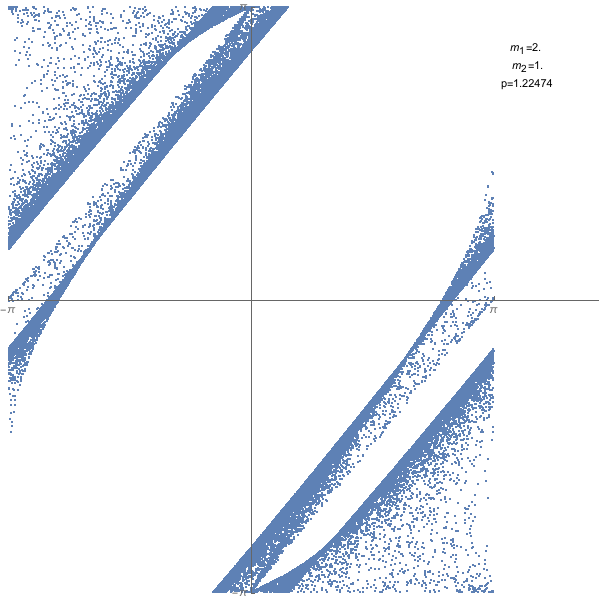}  \\
        $p^2<(m_1-m_2)^2$ & $(m_1-m_2)^2 < p^2 < (m_1+m_2)^2$  \\ \vspace{1em}\\
        \includegraphics[width=5cm]{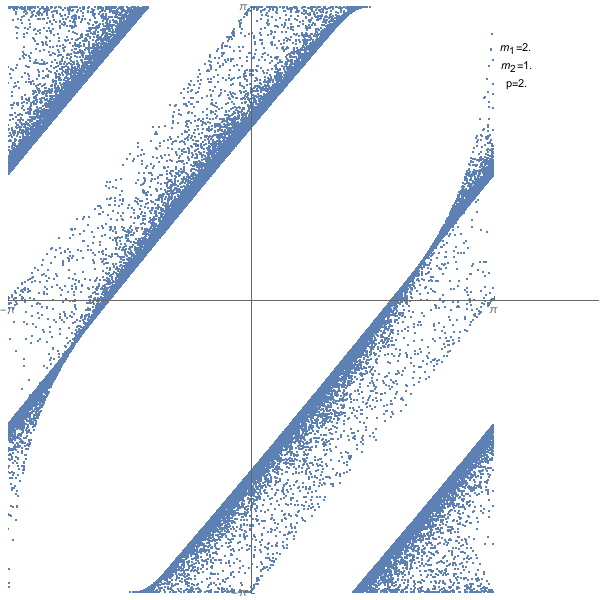}  & \includegraphics[width=5cm]{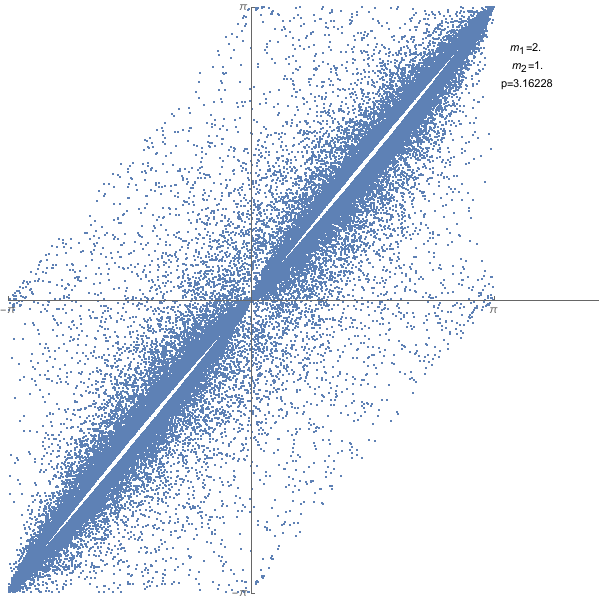} \\
        $(m_1-m_2)^2 < p^2 < (m_1+m_2)^2$ & $p^2>(m_1+m_2)^2$ \\
      \end{tabular}
    \end{center}
    \caption{Example to illustrate the behaviour of the coamoeba. Depicted is the coamoeba $\mathcal C_\Gg$ of the $1$-loop bubble graph $\Gg = x_1+x_2 + (-p^2+m_1^2+m_2^2)x_1 x_2 + m_1^2 x_1^2 + m_2^2 x_2^2$, which has two thresholds $p^2 = (m_1 \pm m_2)^2$. We draw the coamoeba for different kinematic regions with ascending momentum $p^2$. For every crossing of a singularity, the structure of the coamoeba changes. Note, that for the pseudo threshold $p^2 =(m_1-m_2)^2$ the origin is not included in the closure of the coamoeba, whereas for the normal threshold $p^2=(m_1+m_2)^2$ the origin is included. The graphs were produced by sampling random points with Mathematica \cite{WolframResearchIncMathematicaVersion12}.}
    \label{fig:coamoebaBubble}
    \end{figure}
    
    We are mainly interested in the complement of the closure of the coamoeba $\Cc:=\mathbb T^n\setminus\overline{\mathcal C_f}$. It is well known \cite{NisseGeometricCombinatorialStructure2009}, that this complement $\Cc$ is structured into a finite number of connected, convex components and we will denote such a connected component by $\Theta$. Moreover, the number of connected components of $\Cc$ will be bounded by $\vol (\Newt (f))$. We want to stress out, that the integral (\ref{eq:EulerMellinTheta}) only depends on a choice of a connected component $\Theta$, as one can see simply by a homotopic deformation of the integration contour. Thus, for all $\theta\in\mathbb T^n$ inside the same connected component $\Theta\subset \Cc$, the value of $\mathscr M^\theta_f(s,t)$ stays the same, whereas the value may change for another component $\Theta^\prime$. This implies also, that we can avoid the limit in the $i\epsilon$ prescription by choosing a small (but not infinitesimally small) value for $\epsilon$. By choosing another connected component $\Theta$ we will therefore obtain the values of another branch of the Feynman integral. We refer to \cite{HwaHomologyFeynmanIntegrals1966} for the relation between different contours and branches of integrals. However, we have to remark that the coamoeba can not generate the full fundamental group, as the coamoeba can not distinguish between poles of the integrand having the same argument.
    
    However, the coamoeba is sufficient to describe the discontinuity of Feynman integrals in the sense of \cite{CutkoskySingularitiesDiscontinuitiesFeynman1960} as a difference between two $\theta$-analogue Euler-Mellin integrals
    \begin{align}
      \operatorname{Disc} I_\Gamma(\nuu,z) =  I_\Gamma^\epsilon (\nuu,z) - I_\Gamma^{-\epsilon}(\nuu,z)
    \end{align}
    where $\epsilon > 0$ is an adequate small (but not infinitesimally small) real number. \\

    For a practical usage there are different ways to approximate coamoebas. Restricting us to the faces of dimension $1$ we obtain the \textit{shell} of $f$
    \begin{align}
      \mathcal H_f = \bigcup_{\substack{\tau\subseteq\Newt(f) \\ \dim \tau = 1}} \mathcal C_{f_\tau} \point
    \end{align}
    Each cell of the hyperplane arrangement of $\mathcal H_f$ contains at most one connected component of $\Cc$ \cite{ForsgardHypersurfaceCoamoebasIntegral2012}. Thus, $\mathcal H_f$ carries the rough structure of $\mathcal C_f$. Another way, for approximating the coamoeba, is the so-called lopsided coamoeba $\mathcal L \mathcal C_f$ and we have $\mathcal C_f \subseteq \mathcal L \mathcal C_f$. The lopsided coamoeba can be calculated by the set of trinomials $\operatorname{Tri}(f)$ which can be formed by removing all but three monomials from $f$ \cite{ForsgardHypersurfaceCoamoebasIntegral2012}
    \begin{align}
      \overline{\mathcal L \mathcal C}_f = \bigcup_{g\in\operatorname{Tri}(f)} \overline{\mathcal C}_g \point
    \end{align}
    However, there are more effective algorithms to determine lopsided coamoebas. We refer to \cite{ForsgardHypersurfaceCoamoebasIntegral2012, ForsgardTropicalAspectsReal2015} for definition and more detailed discussion. \\
    
    We want to close this section by a result of \cite{BerkeschEulerMellinIntegrals2013}, which guarantees the analytic continuation with respect to the kinematic variables $z$.
    \begin{theorem}[analytic continuation \cite{BerkeschEulerMellinIntegrals2013}]
      Let $\Sing(\mathcal M_\Aa(\nuu)) = \mathbf V( E_A(\Gg)) \subset \mathbb C^N$ be the singular locus of $\mathscr M_\Gg^\Theta$, $z\in\mathbb C^N\setminus\Sing(\mathcal M_\Aa(\nuu))$ a point outside the singular locus and $\Theta$ a connected component of $\mathbb T^n\setminus\overline{\mathcal C_\Gg}$. Then $\mathscr M_\Gg^\Theta(\nuu,z)$ has a multivalued analytic continuation to $\mathbb C^{n+1}\times (\mathbb C^N\setminus\Sing(\mathcal M_\Aa(\nuu))$, which is everywhere $\Aa$-hypergeometric.
    \end{theorem} 
    
    Hence, we can continue the generic Feynman integral analytically for its parameters $\nuu$ as well as for its variables $z$. This result can also be transferred to the non-generic Feynman integral with its physical restrictions on variables $z$, due to the same convergence considerations as in the previous section \ref{sec:FeynmanIntegrals}.

  \section{Non-trivial example: double-edged triangle graph} \label{sec:example}
  
    \begin{wrapfigure}{r}{0.4\textwidth} \label{fig:duncescap}
      \begin{center}
        \includegraphics[trim=0 0 0 1.5cm, clip]{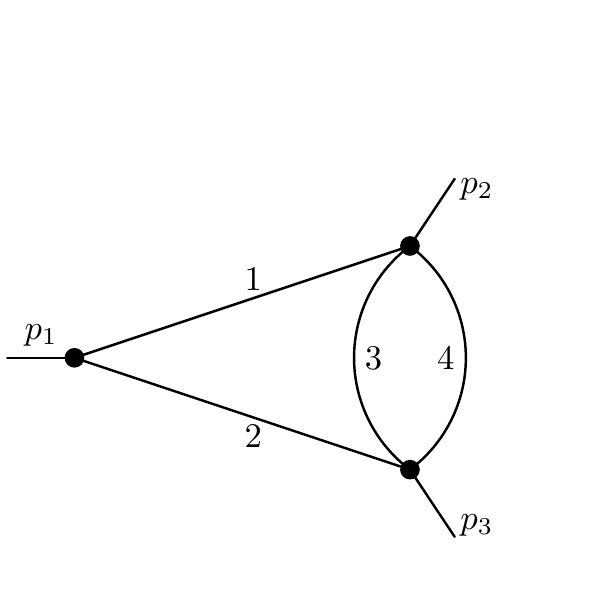}
        \caption{Double-edged triangle graph or dunce's cap graph}
      \end{center}
    \end{wrapfigure}
    In order to demonstrate the methods described before, we want to calculate the leading Landau variety of a massive $2$-loop $3$-point function according to figure \ref{fig:duncescap}, which is also known as ``double-edged triangle graph'' or as ``dunce's cap''. To our knowledge this Landau variety was not published before. Using \HKP we will give the Landau variety in a parametrized form. Compared with the standard methods of eliminating variables from the Landau equations, the \HKP can be calculated very fast. We have to mention, that the reason for the effectiveness of the \HKP lies in a different representation of the result. To determine the defining polynomial of the Landau variety from the parametrized form is still a very time-consuming task. However, for many approaches the parametrized form of Landau varieties can be even more convenient, since the parametrization specifies the Landau singularities directly. \\
    
    The leading Landau variety is nothing else than the hypersurface $\{\Delta_{\Aa}(\Ff) = 0\}$. For the Feynman graph in figure \ref{fig:duncescap} the second Symanzik polynomial is given by
    \begin{align}
      \Ff &= s_1 x_1 x_2 (x_3 + x_4) + s_2 x_1 x_3 x_4 + s_3 x_2 x_3 x_4 \nonumber \\
      & + m_1^2 x_1^2 (x_3 + x_4) + m_2^2 x_2^2 (x_3 + x_4) + m_3^2 x_3^2 (x_1 + x_2 + x_4) + m_4^2 x_4^2 (x_1 + x_2 + x_3)
    \end{align}
    where we abbreviate $s_1 = - p_1^2 + m_1^2 + m_2^2$, $s_2= - p_2^2 + m_1^2 + m_3^2 + m_4^2$ and $s_3 = - p_3^2 + m_2^2 + m_3^2 + m_4^2$. Thus, we can read off the matrix $\Aa$
    {\footnotesize \begin{align}
      \Aa = \left(
             \begin{array}{cccccccccccccc}
               1 & 1 & 1 & 0 & 2 & 2 & 0 & 0 & 1 & 0 & 0 & 1 & 0 & 0 \\
               1 & 1 & 0 & 1 & 0 & 0 & 2 & 2 & 0 & 1 & 0 & 0 & 1 & 0 \\
               1 & 0 & 1 & 1 & 1 & 0 & 1 & 0 & 2 & 2 & 2 & 0 & 0 & 1 \\
               0 & 1 & 1 & 1 & 0 & 1 & 0 & 1 & 0 & 0 & 1 & 2 & 2 & 2 \\
             \end{array} \right) \comma
    \end{align}}
    whose column vectors form the exponents of $\Ff$. We have several possibilities to choose a Gale dual of $\Aa$, e.g.
    {\footnotesize \begin{align}
      \mathcal B = \left(
             \begin{array}{cccccccccc}
               1 & 1 & 1 & 0 & -1 & -1 & 0 & -1 & 0 & -1 \\
               0 & -1 & -1 & 1 & 1 & 1 & -1 & 0 & -1 & 0 \\
               -1 & 0 & -1 & -1 & 0 & -1 & 1 & 1 & -1 & -1 \\
               -1 & -1 & 0 & -1 & -1 & 0 & -1 & -1 & 1 & 1 \\
               0 & 0 & 0 & 0 & 0 & 0 & 0 & 0 & 0 & 1 \\
               0 & 0 & 0 & 0 & 0 & 0 & 0 & 0 & 1 & 0 \\
               0 & 0 & 0 & 0 & 0 & 0 & 0 & 1 & 0 & 0 \\
               0 & 0 & 0 & 0 & 0 & 0 & 1 & 0 & 0 & 0 \\
               0 & 0 & 0 & 0 & 0 & 1 & 0 & 0 & 0 & 0 \\
               0 & 0 & 0 & 0 & 1 & 0 & 0 & 0 & 0 & 0 \\
               0 & 0 & 0 & 1 & 0 & 0 & 0 & 0 & 0 & 0 \\
               0 & 0 & 1 & 0 & 0 & 0 & 0 & 0 & 0 & 0 \\
               0 & 1 & 0 & 0 & 0 & 0 & 0 & 0 & 0 & 0 \\
               1 & 0 & 0 & 0 & 0 & 0 & 0 & 0 & 0 & 0 \\
             \end{array}\right) \point
    \end{align}}
      
    According to section \ref{sec:Adisc} a Gale dual directly gives rise to a parametrization of generic Feynman integrals. Let $z_1,\ldots,z_{14}$ be the coefficients of a generic Feynman integral, i.e. coefficients of the second Symanzik polynomial $\Ff$. Then, the generic discriminant hypersurface $\{\Delta_{\Aa_\Ff}(\Ff)=0\}$ is parametrized by $t\in\mathbb P_{\mathbb C}^{10-1}$
    \begin{align}
      y_1 &= \frac{z_1 z_{14}}{z_3 z_4} = \frac{R_1 t_1}{R_3 R_4} \quad,\qquad    
      y_2 = \frac{z_1 z_{13}}{z_2 z_4} = - \frac{R_1 t_2}{R_2 R_4} \quad,\qquad  
      y_3 = \frac{z_1 z_{12}}{z_2 z_3} = - \frac{R_1 t_3}{R_2 R_3} \nonumber\\
      y_4 &= \frac{z_2 z_{11}}{z_3 z_4} = \frac{R_2 t_4}{R_3 R_4} \quad,\qquad 
      y_5 = \frac{z_2 z_{10}}{z_1 z_4} = - \frac{R_2 t_5}{R_1 R_4} \quad,\qquad 
      y_6 = \frac{z_2 z_9}{z_1 z_3} = - \frac{R_2 t_6}{R_1 R_3} \nonumber\\
      y_7 &= \frac{z_3 z_8}{z_2 z_4} = \frac{R_3 t_7}{R_2 R_4} \quad,\qquad 
      y_8 = \frac{z_3 z_7}{z_1 z_4} = \frac{R_3 t_8}{R_1 R_4} \quad,\qquad 
      y_9 = \frac{z_4 z_6}{z_2 z_3} = \frac{R_4 t_9}{R_2 R_3} \nonumber\\
      y_{10} &= \frac{z_4 z_5}{z_1 z_3} = \frac{R_4 t_{10}}{R_1 R_3} \label{eq:duncescapHKPgeneric}
    \end{align}
    where
    \begin{align}
      R_1 &= t_1+t_2+t_3-t_5-t_6-t_8-t_{10} \quad,\qquad   
      R_2 = -t_2-t_3+t_4+t_5+t_6-t_7-t_9 \\
      R_3 &= t_1+t_3+t_4+t_6-t_7-t_8+t_9+t_{10} \quad,\qquad   
      R_4 = t_1+t_2+t_4+t_5+t_7+t_8-t_9-t_{10} \nonumber \point
    \end{align} % changed signs in R_3 and R_4
    
    However, up to now this \HKP gives the zero locus of discriminants for polynomials with generic coefficients. Thus, in order to adjust this parametrization to the physical relevant case we have to include the constraints given by the relations between coefficients of the Symanzik polynomials. This can be accomplished e.g. by Mathematica \cite{WolframResearchIncMathematicaVersion12} using the command ``Reduce'' or by Macaulay2 \cite{GraysonMacaulay2SoftwareSystem}. For algorithmical reasons it can be more efficient to reduce the effective variables $y_1,\ldots,y_{10}$, the linear forms $R_1,\ldots,R_4$ and the parameters $t_1,\ldots,t_{10}$ step by step. In doing so, the Landau variety splits into two parts
    \begin{align}
      \frac{m_2^2}{m_1^2} &= \frac{t_6^2 t_8}{t_5^2 t_{10}} \quad,\qquad 
      \frac{m_3^2}{m_1^2} = -\frac{t_6^2 t_9}{\left(t_5+t_6\right) \left(t_9+t_{10}\right)t_{10} } \quad,\qquad 
      \frac{m_4^2}{m_1^2} = -\frac{t_3 t_6 t_{10}}{\left(t_5+t_6\right) \left(t_9+t_{10}\right)t_9 } \nonumber\\
      \frac{s_1}{m_1^2} &= -\frac{t_5 \left(t_6 t_9+t_3 t_{10}\right)+t_6 \left(t_6 t_9+t_8 t_9+t_3 t_{10}+t_9t_{10}\right)}{t_5 t_9 t_{10}} \nonumber \\
      \frac{s_2}{m_1^2} &= -\frac{t_5 \left(t_9+t_{10}\right) \left(t_6 t_9+t_3 t_{10}\right) + t_6 \left[t_6 t_9^2+t_8 t_9 \left(t_9+t_{10}\right) -t_{10} \left(t_9^2+t_{10} t_9-t_3 t_{10}\right)\right]}{\left(t_5+t_6\right)  \left(t_9+t_{10}\right)t_9 t_{10}} \nonumber \\
      \frac{s_3}{m_1^2} &= -\frac{t_6 \left[t_6 \left(t_9+t_{10}\right) \left(t_6 t_9-t_8 t_9+t_3 t_{10} + t_9 t_{10}\right)+t_5 \left(t_6 t_9^2+t_3 t_{10}^2\right)\right]}{\left(t_5+t_6\right) \left(t_9+t_{10}\right) t_5 t_9 t_{10} } \label{eq:duncesCapResult1}
    \end{align}
    and 
    \begin{align}
      \frac{m_2^2}{m_1^2} &=  \frac{t_6^2 t_8}{t_5^2 t_{10}} \quad,\qquad 
      \frac{m_3^2}{m_1^2} = \frac{t_6^2 t_9}{t_4 t_{10}^2} \quad,\qquad 
      \frac{m_4^2}{m_1^2} = \frac{t_6^2}{t_4 t_9} \quad,\qquad
      \frac{s_1}{m_1^2} = \frac{t_6 \left(\left(t_4-t_9\right) t_{10} -t_8 t_9\right)}{t_5 t_9 t_{10}} \nonumber \\
      \frac{s_2}{m_1^2} &= -\frac{t_6 \left[t_{10} \left(t_4 \left(t_9+t_{10}\right)+t_9 \left(2 t_6+t_9+t_{10}\right)\right)-t_8 t_9 \left(t_9+t_{10}\right)\right]}{t_4 t_9 t_{10}^2} \nonumber\\
      \frac{s_3}{m_1^2} &= -\frac{t_6^2 \left[t_8 t_9 \left(t_9+t_{10}\right)+t_{10} \left(t_4 \left(t_9+t_{10}\right)-t_9 \left(-2 t_5+t_9+t_{10}\right)\right)\right]}{t_4 t_5 t_9 t_{10}^2} \point \label{eq:duncesCapResult2}
    \end{align}
    Hence, the leading Landau variety of the dunce's cap graph will be given by the values of (\ref{eq:duncesCapResult1}) and (\ref{eq:duncesCapResult2}) for all values $t\in\mathbb P^{6-1}_{\mathbb C}$. We dispense with renaming of the parameters $t$ in order to ensure the reproducibility of the results from (\ref{eq:duncescapHKPgeneric}).

  \section{Conclusion and Outlook}
  
    We have propounded the kinematic singularities of scalar Feynman integrals from the perspective of $\Aa$-hypergeometric theory. This point of view provides a mathematically rigorous description of those singularities by means of principal $A$-determinants, which are polynomials in the coefficients of Symanzik polynomials. More precisely, it turns out that the singular locus of Feynman integrals is the variety defined by the principal $A$-determinant of the sum of the Symanzik polynomials $\Gg=\Uu+\Ff$
    \begin{align}
      \Sing(\mathcal M_\Aa(\nuu)) = \mathbf V (E_A(\Gg)) = \mathbf V \Big(\Delta_A(\Gg) \cdot E_{A_\Uu}(\Uu) \cdot E_{A_\Ff}(\Ff) \cdot R\Big) \point
    \end{align}
    This principal $A$-determinant factorizes in several $A$-discriminants, each one corresponding to a face of the Newton polytope $\Newt(\Gg)$. Hence, we can group those $A$-discriminants to four different partitions. The $A$-discriminant of the full polytope $\Delta_A(\Gg)$ can be identified with the second-type singularities. The two facets $\Newt(\Uu)$ and $\Newt(\Ff)$ of $\Newt(\Gg)$ correspond to a physically non-relevant part $E_{A_\Uu}(\Uu)$ and the so-called Landau variety $E_{A_\Ff}(\Ff)$, respectively. All remaining $A$-discriminants constitute a polynomial $R$, which can be associated to second-type singularities of subgraphs. We have to mention that the factorization may include parts of the singular locus, which were not included in previous approaches. This is due to the fact that, except for simple Feynman graphs, not all truncated polynomials have an equivalent polynomial coming from a subgraph. It is an interesting question for future research if those additional parts will result in a non-trivial contribution. However, these additional parts would appear beyond $1$-loop graphs and banana graphs.

    Apart from the description of the singular locus, we also introduced a powerful tool to determine the singular locus: the \HKp. Thus, the calculation of a Gale dual is sufficient to obtain a parametrization of the hypersurface defined by $A$-discriminants. Clearly, such a representation of a variety differs from a representation via defining polynomials. However, such a representation can be even more convenient for many approaches, as we describe the singularities directly. Also having in mind, that a representation of Landau varieties by a defining polynomial will be an incommensurable effort for almost all Feynman graphs, we want to advertise the usage of \HKp.
    
    In order to study the monodromy of Feynman integrals, we introduced an Euler-Mellin integral with a rotated integration contour. This integral is the analogue to Feynman's $i\epsilon$ prescription and its behaviour is substantially determined by the coamoeba of the polynomial $\Gg$. From the shape of the coamoeba of $\Gg$, we can also conclude the nature of the singularities of Feynman integrals, e.g. we can distinguish between normal/anomalous thresholds and pseudo thresholds. We sketched also several ways to approximate the coamoeba in order to derive efficient algorithms. However, the application of coamoebas to Feynman integrals leaves many questions open and will surely be a worthwhile focus for future research.
    
    \pagebreak
    \appendix
  
  \section{A short guide to Macaulay2}
  
    For the calculation of $A$-discriminants and related objects it will be convenient to use a program, which is specialised on algebraic geometry e.g. Macaulay2 \cite{GraysonMacaulay2SoftwareSystem}. There are two additional libraries for the calculation of classical discriminants and resultants \cite{StaglianoPackageComputationsClassical2018} and $A$-discriminants and $A$-resultants \cite{StaglianoPackageComputationsSparse2020}. In order to calculate also (simple) principal $A$-determinants from the libraries \cite{StaglianoPackageComputationsClassical2018,StaglianoPackageComputationsSparse2020}, we will present an elementary package \verb|Landau.m2| below, which is adjusted to the approach of Feynman integrals. 
    
    In order to demonstrate the usage, we will show the calculation of the Landau variety of the triangle graph.
    {\tiny
    \begin{verbatim}
$ M2 --no-preload
Macaulay2, version 1.18

i1 : loadPackage "Landau";

i2 : QQ[s1,s2,s3,b1,b2,b3][x1,x2,x3]; U=x1+x2+x3; F = s1*x2*x3 + s2*x1*x3 + s3*x1*x2 + b1*x1^2 + b2*x2^2 + b3*x3^2;

i5 : principalAdet(U+F)

        2              2                       2                                                                                                                          
o5 = {s1  - 2s1*s2 + s2  - 2s1*s3 - 2s2*s3 + s3  + 4s1*b1 + 4s2*b2 - 4b1*b2 + 4s3*b3 - 4b1*b3 - 4b2*b3, - s1 + b2 + b3, 1, 1, - s2 + b1 + b3, 
     -----------------------------------------------------------------------------------------------------------------------------------------
                                         2                              2       2       2                                 2                  2
     1, b3, - s3 + b1 + b2, 1, 1, 1, - s1  + 4b2*b3, b2, - s1*s2*s3 + s1 b1 + s2 b2 + s3 b3 - 4b1*b2*b3, 1, 1, 1, b1, - s2  + 4b1*b3, 1, - s3  + 4b1*b2}
    \end{verbatim}}

    Package \verb|Landau.m2|
    {\tiny
    \begin{Verbatim}[frame=single]
-- Landau - small package to calculate Landau varieties, i.e. the principal A-determinant, mostly fitted to Feynman integrals 
-- Instructions:
-- the packages Polyhedra and Resultants have to be installed. If they are not installed, run: installPackage "Polyhedra"
-- this file have to be stored in a Macaulay2 path. The Macaulay2 paths can be displayed by the command: path
-- this package can be used by the command: loadPackage "Landau"

newPackage(
        "Landau",
        Version => "2.1", 
        Date => "August 24, 2021",
        Authors => {{Name => "Rene P. Klausen", 
                  Email => "klausen@physik.hu-berlin.de"}},
        Headline => "Calculating Landau varieties by means of principal A-determinants",
        DebuggingMode => true
        )
        
export {"principalAdet","generalDiscriminant","allTruncs"}

needsPackage "Polyhedra"
needsPackage "Resultants"
needsPackage "SparseResultants"

-- general stuff
ListTimes = (L1,L2) -> apply(L1,L2, (i,j) -> i*j)

-- truncation of polynomials
poly2A = f -> transpose(matrix(exponents(f)))
ptsOfFace = (A,face) -> apply(splice({0..numgens(source(A))-1}),i -> if contains(face,convexHull(submatrix(A,{i})))==true then 1 else 0 )    
faceTruncation = (f,A,face) -> sum(ListTimes(terms(f),ptsOfFace(A,face)))
truncatedPolynomial = (f,k) -> (A := poly2A(f); P := convexHull(A); face := facesAsPolyhedra(k,P); for i from 0 to #face-1 
      list faceTruncation(f,A,face_i)) 
allTruncs = f -> (n := numgens(ring(f)); mingle delete({},for i from 0 to n list truncatedPolynomial(f,i)))

-- fit rings
usedVars = (f,R) -> (gR := gens(R); delete("del",for i from 0 to #gR-1 list if diff(gR_i,f)==0 then "del" else i ))
fitRing = f -> (R := ring(f); substitute(f,first(selectVariables(usedVars(f,R),R))))
factorOut = (f,var) -> if pseudoRemainder(f,var)==0 then f//var else f
completeFactorOut = (f,varList) -> ((for i from 0 to #varList-1 do f = factorOut(factorOut(f,varList_i),varList_i)); f) 
      -- at most for quadratic expressions (as in Symanzik polynomials)
dehomogenize = f -> (if isHomogeneous(f) then (sub(f,last(gens(ring f))=>1) ) else f);

--principal A determinant
generalDiscriminant = f -> (try (
  m:= # terms f;
  if m == 1 then (
    print("vertex type"); 
    coeff := ((coefficients(f))_1)_(0,0);
    sub(coeff, coefficientRing ring f) )
  else (
    f = fitRing f;
    n:= numgens ring f;
    if m-1 <= n then (
      print("dense discriminant"); 
      f = fitRing dehomogenize f; 
      f = completeFactorOut(f, gens ring f); 
      f = fitRing dehomogenize f;
      denseDiscriminant f )
    else (
      print("sparse discriminant"); 
      sparseDiscriminant f) )
  ) else (print("NN"); "NN")
);
principalAdet = f -> apply(allTruncs(f),generalDiscriminant);

beginDocumentation()
document { 
        Key => Landau,
        Headline => "Calculating Landau varieties by means of principal A-determinants",
        EM "Landau", " is a basic package to calculate Landau singularities."
        }
document {
        Key => {allTruncs},
        Headline => "all truncated polynomials",
        Usage => "allTruncs(f)",
        Inputs => {"a polynomial f"},
        Outputs => {"a list of all truncated polynomials"},
        EXAMPLE lines ///
           QQ[m1,m2,s][x1,x2]; F = m1*x1^2 + m2*x2^2 + s*x1*x2;
           allTruncs F
        ///
        }
document {
        Key => {generalDiscriminant},
        Headline => "calculate the A-discriminant",
        Usage => "generalDiscriminant f",
        Inputs => {"f a polynomial"},
        Outputs => {"A-discriminant"},
        EXAMPLE lines ///
           QQ[m1,m2,s][x1,x2]; F = m1*x1^2 + m2*x2^2 + s*x1*x2;
           generalDiscriminant F
        ///
        }     
document {
        Key => {principalAdet},
        Headline => "calculate the simple principal A-determinant",
        Usage => "principalAdet f",
        Inputs => {"f a polynomial"},
        Outputs => {"A list of all A-discriminants of all truncated polynomials, the (simple) principal A-determinant is the product of all 
            list elements. Additionally the used method for every A-discriminant is printed on screen."},
        EXAMPLE lines ///
           QQ[m1,m2,s][x1,x2]; F = m1*x1^2 + m2*x2^2 + s*x1*x2;
           principalAdet F
        ///
        }             
end--
    \end{Verbatim}
    }

    \printbibliography

\end{document}